\documentclass[oneside,reqno]{amsart}

\usepackage{amsmath}
\usepackage{amssymb}
\usepackage{amsfonts}       
\usepackage{nicefrac}       
\usepackage{microtype}      
\usepackage{bm}
\usepackage{dsfont}
\usepackage{color,soul}
\usepackage[round]{natbib}
\usepackage{tabularx}
\usepackage{graphicx}
\usepackage{caption}
\usepackage{subcaption}
\usepackage[margin=3cm]{geometry}
\usepackage[onehalfspacing]{setspace}
\usepackage[hidelinks]{hyperref}  	    
\usepackage[ruled,vlined,noline]{algorithm2e}
\usepackage{multirow}
\usepackage{booktabs}

\newtheorem{theorem}{Theorem}[section]
\newtheorem{lemma}[theorem]{Lemma}
\newtheorem{corollary}[theorem]{Corollary}

\theoremstyle{definition}

\newtheorem{example}{Example}[section]

\theoremstyle{remark}

\newcommand{\Nd}{\mathcal{N}}
\newcommand{\R}{\mathbb{R}}

\newcommand{\G}{\mathcal{G}}

\newcommand{\e}{\varepsilon}
\newcommand{\E}{\mathbb{E}}
\newcommand{\Cov}{\operatorname{Cov}}
\newcommand{\Var}{\operatorname{Var}}

\newcommand*\diff{\mathop{}\!\mathrm{d}}

\newcommand{\1}{\mathds{1}}

\DeclareMathOperator{\eff}{eff}

\newcommand{\Or}{\mathcal{O}}
\DeclareMathOperator{\diag}{diag}
\newcommand{\x}{\boldsymbol{x}}

\newcommand{\z}{\boldsymbol{z}}
\newcommand{\tht}{\boldsymbol{\theta}}
\newcommand{\Th}{\boldsymbol{\Theta}}
\newcommand{\X}{\boldsymbol{X}}
\newcommand{\Z}{\boldsymbol{Z}}

\newcommand{\mx}{\bm{\mu}}
\newcommand{\B}{\bm{\beta}}
\newcommand{\f}{\mathbf{f}}
\newcommand{\M}{\mathbf{M}}
\newcommand{\m}{\boldsymbol{m}}
\newcommand{\Q}{\mathbf{Q}}
\newcommand{\Rho}{\mathbf{P}}
\newcommand{\g}{\boldsymbol{g}}

\renewcommand{\E}{\operatorname{E}}
\newcommand{\Prob}{\operatorname{P}}

\newcommand{\MSE}{\operatorname{MSE}}
\newcommand{\Id}{\bm{\mathbb{I}}}

\numberwithin{equation}{section}

\begin{document}
	\title{$D$-Optimal Subsampling Design for Multiple Linear Regression on Massive Data}
	
	
	\author{Torsten Glemser}
	\thanks{Corresponding author: Torsten Glemser. \textit{E-mail address}: \texttt{torsten.reuter@ovgu.de}}
	\address{Torsten Glemser. Otto von Guericke University Magdeburg. Universitätsplatz 2, 39106 Magdeburg,
		Germany}
	\curraddr{}
	\email{torsten.reuter@ovgu.de}
	
	\author{Rainer Schwabe}
	\address{Rainer Schwabe. Otto von Guericke University Magdeburg. Universitätsplatz 2, 39106 Magdeburg,
		Germany}
	\curraddr{}
	\email{rainer.schwabe@ovgu.de}

	\date{}
		
\maketitle

\begin{abstract}
	Data reduction is a fundamental challenge of modern technology, where classical statistical methods are not applicable because of computational limitations.  
	We consider multiple linear regression for an extraordinarily large number of observations, but only a few covariates. 
	Subsampling aims at the selection of a given proportion of the existing original data. 
	Under distributional assumptions on the covariates, we derive $D$-optimal subsampling designs and study their theoretical properties. 
	We make use of fundamental concepts of optimal design theory and an equivalence theorem from constrained convex optimization. 
	The thus obtained subsampling designs provide simple rules for whether to accept or reject a data point, allowing for an easy algorithmic implementation.
	In addition, we propose a simplified subsampling method with lower computational complexity that deviates from the $D$-optimal design.
	We present a simulation study, comparing both subsampling schemes with the IBOSS method in the case of a fixed size of the subsample. 
\end{abstract}

\section{Introduction}
\label{sec:intro}
Data reduction is a fundamental challenge of modern technology, which allows us to collect huge amounts of data. 
Often, technological advances in computing power do not keep pace with the amount of data, creating a need for data reduction.  
We speak of big data whenever the full data size is too large to be handled by traditional statistical methods. 
In this paper, we consider the case of so-called massive data where the number of units is extremely large, while the number of covariates is relatively small. 
Subsampling for high-dimensional data is studied e.g.~in the work of \cite{singh2023subdata}, which combines LASSO and subsampling.
To deal with huge amounts of units one of two methods is used: 
one strategy is to divide the data into several smaller data sets and compute them separately, known as divide-and-conquer, see \cite{lin2011aggregated}. 
Alternatively, one can find an informative subsample of the full data. 
This can be done in a probabilistic way, where units are sampled according to
some sampling distribution. 
\cite{ma2014statistical} present subsampling methods for linear regression models called algorithmic leveraging. There, the sampling distribution is based on the normalized statistical leverage scores of the covariate matrix.
Volume sampling, where subsamples are chosen proportional to the squared volume of the parallelepiped spanned by its units, is studied by \cite{derezinski2018reverse}.
On the other hand, subdata can be selected using a deterministic method.
\cite{shi2021model} present a space-filling subsampling method that is deterministic. 
There, the minimal distance between two units in the subdata is maximized. 
Most prominently, \cite{wang2019information} have introduced the information-based optimal subdata selection (IBOSS) to tackle big data linear regression in a deterministic fashion based on $D$-optimality. 
The IBOSS approach selects the outer-most data points of each covariate successively. 
Other subsampling methods for linear regression include the works by
\cite{wang2021orthogonal}, who have introduced orthogonal subsampling inspired by orthogonal arrays, which selects units in the corners of the design space and the optimal design based subsampling scheme by \cite{deldossi2021optimal}.
Subsampling becomes increasingly popular, leading to more work outside linear models.
\cite{cheng2020information} extend the idea of the IBOSS method from the linear model to logistic regression and other work on generalized linear regression includes the papers by \cite{zhang2021optimal} and \cite{ul2019optimal}.
Recently, \cite{su2022two-stage}  consider subsampling for missing data, whereas \cite{joseph2021supervised} focus on nonparametric models and make use of the information in the dependent variables. 
Various works consider subsampling when the full data is distributed over several data sources, among them \cite{yu2022optimal} and \cite{zhang2021distributed}
For a more thorough recent review on design inspired subsampling methods see the work by \cite{yu2023review}.

In this paper, we assume that both the regression model and the shape of the joint distribution of the covariates are known.
We search for $D$-optimal continuous subsampling designs of total measure $\alpha$ that are bounded from above by the distribution of the covariates.  
\cite{wynn1977optimum} and \cite{fedorov1989optimal} were the first to study such directly bounded designs. 
More recent work includes the paper by \cite{pronzato2004minimax} and more recently \cite{pronzato2021sequential} in the context of sequential subsampling. 
In \cite{reuter2023optimal} we study bounded $D$-optimal subsampling designs for polynomial regression in one covariate, using similar ideas as we use here.

In the present work, we extend results in \cite{reuter2023optimal} to the situation of multiple covariates.
In contrast to other work, we stay with the unstandardized version of the design emphasizing the subsampling character of the design. 
For the characterization of the optimal subsampling design, we make use of an equivalence theorem given in  \cite{sahm2001note}. 
This equivalence theorem allows us to construct subsampling designs for different settings of the distributional assumptions on the covariates. 
Based on this, we propose a simple subsampling scheme for selecting units. 
The resulting selection method includes all data points in the support of the optimal subsampling design and rejects all other units. 
Although this approach is basically probabilistic, as it allows probabilities for selection, the resulting optimal subsampling design is purely deterministic since it depends only on the acceptance region defined by the optimal subsampling design. 
We comment on the asymptotic behavior of the ordinary least squares estimator based on the $D$-optimal subsampling design that selects the data points with the largest Mahalanobis distance from the mean of the data.

Since the proposed algorithm requires computational complexity of the same magnitude as calculating the least squares estimator on the full data, we also propose a simplified version with lower computational complexity, that takes the variance of the covariates into account while disregarding the covariance between them.  

The rest of this paper is organized as follows. 
After introducing the model in Section~\ref{sec:model}, 
we present the setup and establish necessary concepts and notations in Section~\ref{sec:design}.
There, we first illustrate our methodology by 
the example of ordinary linear regression in one covariate. 
Then we construct optimal subsampling designs 
for multiple linear regression in more than one covariate. 
Algorithms are given in Section~\ref{sec:algorithms}
for generating subsamples from a full data set.
In Section~\ref{sec:method}, we consider the case of a fixed subsample size 
and examine the performance of our method in simulation studies. 
All simulations were done using R Statistical Software 
\cite[v4.2.2]{RCoreTeam2020}
and the pseudo-random variates implemented therein.
Finally, we make some concluding remarks 
and discuss some extensions in Section~\ref{sec:discussion}.
Technical details and proofs are deferred to an Appendix.

\section{Model Specification}
\label{sec:model}

We consider the situation of massive data, 
or, more precisely,
of multivariate data $(y_{i}, \x_{i})$, 
when the number $n$ of units $i = 1, \ldots, n$ is very large.
Here,
the response $y_{i}$ is the outcome of a response variable $Y_{i}$ 
and the vector $\x_{i} = (x_{i1}, \ldots, x_{id})^{\top}$
is the realization of the corresponding 
$d$-dimensional random vector $\X_{i}$ of covariates. 
We suppose that the relation between the covariates $\X_{i} = (X_{i1},\ldots,X_{id})^{\top}$
and the response $Y_{i}$
is given by the multiple linear regression model
\begin{equation}
	\label{eq:model-multiple-regression}
	Y_{i} = \beta_{0} + \beta_{1} X_{i1} + \beta_{2}X_{i2} + \ldots + \beta_{d}X_{id} + \e_{i} \, .
\end{equation}
Here $\beta_{0}$ denotes the intercept 
and $\beta_{j}$ is the slope parameter for the $j$th covariate $x_{j}$ in 
the covariates vector $\x = (x_{1}, \ldots, x_{d})^{\top}$, $j = 1, \ldots, d$. 
Our aim is to provide an approach 
which allows to estimate the vector 
$\B = (\beta_{0}, \ldots, \beta_{d})^{\top}$
of regression parameters
as precisely as possible with least possible efforts.

For notational convenience, we write the multiple linear regression model 
\eqref{eq:model-multiple-regression} as a general linear model 
\begin{equation*}
	\label{eq:general-linear-model}
	Y_{i} = \f(\X_{i})^{\top} \B + \e_{i} \, ,
\end{equation*}
$i=1, \ldots, n$,
where $\f(\x) = (1,\x^{\top})^{\top} = (1,x_{1},\ldots,x_{d})^{\top}$
is the $(d + 1)$-dimensional vector of regression functions. 
The observational errors $\e_{i}$ 
are assumed to be uncorrelated and homoscedastic 
with zero mean, $\E[\e_{i}] = 0$, and 
finite variance, $\Var[\e_{i}] = \sigma_{\e}^{2} > 0$.

Further, we assume that the covariates $\X_{i}$ 
are independent and identically distributed 
and have a common continuous multivariate distribution 
with probability density function $f_{\X}$.
The error terms $\e_{1}, \ldots, \e_{n}$
and the covariates $\X_{1}, \ldots, \X_{n}$ are assumed to be independent of each other.

\section{Continuous Subsampling Design}
\label{sec:design}

We consider a scenario where the response $y_{i}$ are expensive to obtain.
Therefore, only a proportion $\alpha$ ($0 < \alpha < 1$) of the $y_{i}$ will be observed
while all values $\x_{i}$ of the covariates are available. 
Alternatively, we may consider that all data $(y_{i}, \x_{i})$ are available, 
but parameter estimation is only computationally feasible 
on a smaller proportion $\alpha$ of the data. 
Either setup leads to the question 
which subsample of the data $(y_{i} \x_{i})$ 
yields the best estimate of the parameter vector $\B$ 
or essential parts of it. 
\vspace{3mm}

\subsection{General Case}
\label{subsec:continuous-subsample-general-case}
\mbox{}

Throughout this section, we assume 
that the distribution of $\X_{i}$ and, hence, its density $f_{\X}$ 
are known in advance
in order to derive theoretical results. 
For given proportion $\alpha$, 
we define a (continuous) subsampling design $\xi$
as a measure on $\R^{d}$ with total mass $\alpha$
which is uniformly bounded by the distribution of $\X$,
i.\,e.\ $\xi(B) \leq \Prob(X_{i} \in B)$ for all measurable sets $B$ on $\R^{d}$.
In particular, for $X_{i}$ with continuous distribution,
also the subsampling design $\xi$ is continuous
and has density function $f_{\xi}$ 
satisfying $\int f_{\xi}(\x)\diff \x = \alpha$ 
and 
$f_{\xi}(\x) \leq f_{\X}(\x)$ for all $\x \in \R^{d}$.

To evaluate the quality of a subsampling design $\xi$,
we use its (unstandardized) information matrix
\begin{equation*}
	\label{eq:information-matrix}
	\M(\xi) = \int \f(\x) \f(\x)^{\top} f_\xi(\x) \diff \x \, .
\end{equation*}
To ensure a meaningful information matrix 
$\M(\xi)$ with finite entries
for any subsampling design $\xi$,
we have to require the existence of
finite second moments ($\E[X_{ij}^{2}] < \infty$)
of the covariates $\X_{i}$.
Note
that, for any continuous subsampling design $\xi$,
the information matrix $\M(\xi)$
is nonsingular (almost surely) 
because $\M(\xi)$ is based on a density $f_{\xi}$ 
which cannot be concentrated on a proper affine subspace of $\R^{d}$.

According to a given subsampling design $\xi$,
a real subsample can be generated from the full data 
by selecting any unit $i$ with probability $f_{\xi}(\x_{i})/f_{\X}(\x_{i})$.
By the Law of Large Numbers, the effective proportion of accepted items 
will tend to $\alpha$ as the size $n$ of the data tends to infinity.
In so far, a subsampling design $\xi$ provides a probabilistic method
for generating a subsample with approximately the prescribed proportion $\alpha$.
In particular, 
approximate uniform random sampling 
can be achieved by a subsampling design $\xi_{\mathrm{unif}}$ 
with density $f_{\xi_{\mathrm{unif}}}(\x) = \alpha f_{X}(\x)$.

For a sampling procedure according to a subsampling design $\xi$,
the least squares estimator $\hat{\B}$ based on the sample 
is asymptotically normal with asymptotic covariance matrix 
proportional to the inverse $\M(\xi)^{-1}$ 
of the information matrix $\M(\xi)$
when the size $n$ of the full data tends to infinity, 
$\sqrt{n}(\hat{\B}_{n} - \B) \stackrel{\mathcal{D}}{\to} 
\Nd_{d+1}\left(\mathbf{0}, \sigma_{\e}^{2} \M(\xi)^{-1}\right)$. 
For details see Lemma~\ref{lemma:asymptotic-normality} in the Appendix.
The information matrix
$\M(\xi)$ thus measures the quality of 
a subsampling design $\xi$
in the sense that the asymptotic covariance becomes smaller
when the information gets larger.
Hence, we aim at maximizing the information
in order to minimize the covariance.

As, in general, the information matrix cannot be maximized 
in the Loewner sense of nonnegative-definiteness,
we adopt here the most popular $D$-criterion 
which aims at maximizing the determinant $\det(\M(\xi))$ 
of the information matrix
or, equivalently, to minimize the determinant of the
asymptotic covariance matrix.
Thus $D$-optimality may be interpreted as minimization 
of the volume of the asymptotic confidence ellipsoid 
of the parameter vector $\B$ based on the least squares estimator $\hat{\B}$. 
The $D$-optimal subsampling design 
of proportion $\alpha$ will be denoted by $\xi_{\alpha}^{*}$. 

The logarithmic version 
$\Phi_D(\xi) = -\ln(\det(\M(\xi)))$
of the $D$-criterion is convex.
Thus, methods from convex optimization may be employed
to characterize a $D$-optimal subsampling design $\xi_{\alpha}^{*}$
\citep[see e.\,g.\ ][Chapter~3]{silvey1980optimal}.
In particular, 
we apply a constrained equivalence theorem
under Kuhn-Tucker conditions 
\citep[][Corollary~1~(c)]{sahm2001note},
see Theorem~\ref{theorem:opt-design} in the Appendix.
For any subsampling design $\xi$, 
we define measures of location and dispersion
\begin{eqnarray}
	\label{eq:concentration ellipsoid-mean-scale}
	\m(\xi) & = & \frac{1}{\alpha} \int \x f_{\xi}(\x) \diff \x
		\qquad \mbox{and}
	\nonumber
	\\
	\mathbf{S}(\xi) & = & \int \x \x^{\top} f_{\xi}(\x) \diff \x - \alpha \m(\xi) \m(\xi)^{\top} \, .
\end{eqnarray}
Then we can characterize 
a $D$-optimal subsampling design $\xi_{\alpha}^{*}$
as follows.

\begin{theorem}
	\label{theorem:opt-design-general-distribution}
	The subsampling design $\xi_{\alpha}^{*}$ is $D$-optimal 
	if and only if
	$\xi_{\alpha}^{*}$ has density 
	\begin{equation*}
		\label{eq:opt-design-general-distribution}
		f_{\xi_{\alpha}^{*}}(\x) 
		= \left\{
				\begin{array}{cl}
					f_{\X}(\x) & \mbox{ for } (\x - \m(\xi_{\alpha}^{*}))^{\top} \mathbf{S}(\xi_{\alpha}^{*})^{-1} (\x - \m(\xi_{\alpha}^{*})) \geq c \, ,
				\\
				0 & \mbox{ otherwise,}
			\end{array} 
		\right.  
	\end{equation*}
	where $c$ satisfies 
	$\Prob\left((\X_{i} - \m(\xi_{\alpha}^{*}))^{\top} 
							\mathbf{S}(\xi_{\alpha}^{*})^{-1} 
							(\X_{i} - \m(\xi_{\alpha}^{*})) 
					\geq c
				\right) 
		= \alpha$.
\end{theorem}

For the $D$-optimal subsampling design $\xi_{\alpha}^{*}$,
the resulting sampling procedure is deterministic:
units will be selected if their values of the covariates
lie outside the ellipsoid with center $\m(\xi_{\alpha}^{*})$, 
dispersion matrix $\mathbf{S}(\xi_{\alpha}^{*})$,
and ``radius'' $c$ as defined 
in equation~\eqref{eq:concentration ellipsoid-mean-scale}
and Theorem~\ref{theorem:opt-design-general-distribution}.
Units will be not included if they lie in the interior of that ellipsoid.

Note that $c$ is the $(1 - \alpha)$-quantile of the distribution of
$(\X_{i} - \m(\xi_{\alpha}^{*}))^{\top} \mathbf{S}(\xi_{\alpha}^{*})^{-1} (\X_{i} - \m(\xi_{\alpha}^{*}))$.
\vspace{3mm}

\subsection{A Single Covariate}
\label{subsec:continuous-subsample-single-covariate}
\mbox{}

Before we treat the case of multiple linear regression,
we briefly summarize results for the case of ordinary linear regression 
presented in \citet[Section~4]{reuter2023optimal}
where the single covariate $x$ has dimension $d = 1$.
There, the regression function $\f$
and the parameter vector $\B$ reduce 
to $\f(x) = (1, x)^{\top}$ and $\B = (\beta_{0}, \beta_{1})^{\top}$,
respectively. 
We assume that the distribution of the single covariate $\X_{i}$
has finite second moment ($\E[X_{i}^{2}] < \infty$). 
In this situation, the result 
of Theorem~\ref{theorem:opt-design-general-distribution}
simplifies.

\begin{corollary}
	\label{corollary:opt-design-one-cov-general}
	For $d = 1$, 
	the subsampling design $\xi_{\alpha}^{*}$ is $D$-optimal 
	if and only if
	$\xi_{\alpha}^{*}$ has density 
	\begin{equation*}
		\label{eq:opt-design-one-cov-general}
		f_{\xi_{\alpha}^{*}}(x) 
		= \left\{
				\begin{array}{cl}
					f_{X}(x) & \mbox{ for } x \leq a \mbox{ or } x \geq b \, ,
					\\
					0 & \mbox{ otherwise,}
				\end{array} 
		\right.  
	\end{equation*}
	where $(a + b)/2  = \alpha^{-1} \int x f_{\xi_{\alpha}^{*}}(x) \diff x$ 
	and 
	$\Prob(a < X_{i} < b) = 1 - \alpha$.
\end{corollary}

We can make use of symmetry considerations
to further simplify the characterization 
of the $D$-optimal subsampling design.
If the distribution of the covariate $X_{i}$ is symmetric at $0$,
then $\E[X_{i}] = 0$,
and the density is invariant with respect to sign change $g(x) = - x$,
i.\,e.\ $f_{X}(- x) = f_{X}(x)$. 
Moreover, the regression function $\f(x)$ is linearly equivariant
with respect to sign change, 
as 
$\f(g(x)) 
	= \begin{pmatrix}
			1 & 0
			\\
			0 & -1
		\end{pmatrix}
		\f(x)$ 
for all $x$.
For any subsampling design $\xi$,
its image $\xi^{g}$ under sign change,
i.\,e.\ $\xi^{g}(B) = \xi(- B)$ for any measurable set $B$,
is itself a subsampling design
as $\xi^{g}$ has mass $\alpha$ and
$f_{\xi^{g}}(x) = f_{\xi}(-x) \leq f_{X}(x)$
by the symmetry of $f_{X}$.
As a consequence, also the symmetrization 
$\bar{\xi} = (\xi + \xi^{g}) / 2$
is a subsampling design
satisfying $f_{\bar{\xi}}(x) = (f_\xi(x) + f_{\xi^{g}}(x))/2 \leq f_{X}(x)$.
Further, the $D$-criterion is invariant 
with respect to sign change,
i.\,e.\ $\Phi_D(\xi^{g}) = \Phi_D(\xi)$,
so that $\xi$ is dominated by its symmetrization $\bar{\xi}$,
i.\,e.\ $\Phi_D(\bar{\xi}) \leq \Phi_D(\xi)$,
because of the convexity of the $D$-criterion.
Thus we can restrict our search for a $D$-optimal subsampling design $\xi_{\alpha}^{*}$ 
to the essentially complete class of symmetric subsampling designs 
with $f_\xi(-x) = f_\xi(x)$ 
\citep[see][Chapter 13.11]{pukelsheim1993optimal}.

For any invariant subsampling design $ \xi$, 
the off-diagonal entry $\int x f_{\xi}(x) \diff x$ 
of the information matrix $\M(\xi)$
is equal to $0$.
The cut-off points $a$ and $b$ 
in Corollary~\ref{corollary:opt-design-one-cov-general}
are symmetric at $0$, i.\,e.\ $a = - b$,
and can be determined explicitly in terms
of the distribution of $X_{i}$.

\begin{corollary}
	\label{corollary:opt-design-one-cov-0}
	Let $d = 1$ and $f_{X}$ be symmetric at $0$.
	The subsampling design $\xi_{\alpha}^{*}$ is $D$-optimal 
	if and only if
	$\xi_{\alpha}^{*}$ has density 
	\begin{equation*}
		\label{eq:opt-design-one-cov-0}
		f_{\xi_{\alpha}^{*}}(x) 
			= \left\{
					\begin{array}{cl}
						f_{X}(x) & \mbox{ for } |x| \geq x_{1 - \alpha/2} \, ,
						\\
						0 & \mbox{ otherwise,}
					\end{array} 
				\right.  
	\end{equation*}
	where $x_{1 - \alpha/2}$ is the $ (1 - \alpha/2) $-quantile of $ X_{i} $.
\end{corollary}

Under the conditions of 
Corollary~\ref{corollary:opt-design-one-cov-0},
the information matrix $\M(\xi_{\alpha}^{*})$
of the $D$-optimal subsampling design $\xi_{\alpha}^{*}$
is diagonal,
\[
	\M(\xi_{\alpha}^{*})
	= \begin{pmatrix}
			\alpha & 0
			\\
			0 & m_{2}(\xi_{\alpha}^{*})
		\end{pmatrix} \, ,
\]
where $m_{2}(\xi_{\alpha}^{*}) = \int x^2 f_{\xi_{\alpha}^{*}} \diff x$
is the second moment of  $\xi_{\alpha}^{*}$.

By equivariance with respect to location shifts $g(x) = x + \mu$,
this result can be transferred 
to distributions symmetric at some location parameter $\mu$ 
($ f_{X}(\mu - x) = f_{X}(\mu + x) $),
see~\citet[Theorem~3.2]{reuter2023optimal}. 

\begin{corollary}
	\label{corollary:opt-design-one-cov-mu}
	Let $d=1$ and $f_{X}$ be symmetric at $\mu$.
	The subsampling design $\xi_{\alpha}^{*}$ is $D$-optimal 
	if and only if
	$\xi_{\alpha}^{*}$ has density 
	\begin{equation*}
		\label{eq:opt-design-one-cov-mu}
		f_{\xi_{\alpha}^{*}}(x) 
		= \left\{
				\begin{array}{cl}
					f_{X}(x) & \mbox{ for } x \leq x_{\alpha/2} \mbox{ or } x \geq x_{1 - \alpha/2} \, ,
					\\
					0 & \mbox{ otherwise,}
				\end{array} 
		\right.  
	\end{equation*}
	where $x_{\alpha/2}$ and $x_{1 - \alpha/2}$ are 
	the $(\alpha / 2)$- and $ (1 - \alpha / 2) $-quantiles of $X_{i}$, respectively.
\end{corollary}

This procedure can be interpreted as the approximate counterpart to the
IBOSS method proposed by \cite{wang2019information} in one dimension
in which both those $n \alpha / 2$ units are selected
which have the largest values of the covariate  
as well as those $n \alpha / 2$ units 
which have the smallest values of the covariate. 
\vspace{3mm}

\subsection{Multiple Covariates With Elliptical Distribution}
\label{subsec:continuous-subsample-multiple-covariates-elliptical-distribution}
\mbox{}

We now extend the results for a single covariate to
the situation of multiple linear regression where
the covariates vector $\X_{i}$ has dimension $d > 1$.
Motivated by the shape of the support of the 
$D$-optimal subsampling design $\xi_{\alpha}^{*}$
in Theorem~\ref{theorem:opt-design-general-distribution}
and the symmetry property in
Corollary~\ref{corollary:opt-design-one-cov-0},
we start with the case that the multivariate covariates $\X_{i}$ 
have a centered spherical distribution,
i.\,e.\ the density $f_{\X}$ has
spherical contours such that 
$f_{\X}(\x) = f_{0}(\Vert \x \Vert^{2})$
for some univariate function $f_{0}$, 
where $\Vert \x \Vert = (\x^{\top} \x)^{1/2}$ 
denotes the Euclidean norm of the vector $\x$.
When the distribution of the covariates $\X_{i}$ 
is centered and spherical,  
this implies that  $\X_{i}$ has mean $\E[\X_{i}] = \mathbf{0}$
and covariance matrix 
$\Cov[\X_{i}] = \sigma^{2} \Id_{d}$, 
where $\Id_{d}$ denotes the 
identity matrix of dimension $d$.
Moreover, all $d$ single covariates $X_{ij}$ follow the same distribution symmetric at $0$.
The most prominent representative 
of these spherical distributions is the  
standard multivariate normal distribution 
with $\sigma^{2} = 1$ and $f_{0}(t) = (2 \pi)^{-d/2} \exp(-t/2)$.
But also multivariate $t$-distributions are covered.
The sphericity of the distribution of the covariates
provides symmetry properties which allow for a simple
characterization of optimal subsampling designs.
In particular, the distribution is invariant 
with respect to the special orthogonal group $SO(d)$ 
of rotations $\g$ in $\R^{d}$ about the origin $\mathbf{0}$,
i.\,e.\ $f_{\X}(\g(\x)) = f_{\X}(\x)$ for all $\g \in SO(d)$.

To make use of the rotational invariance, 
we characterize subsampling designs $\xi$ 
in their representation in hyperspherical (polar) coordinates, 
where a point $\x$ in $\R^{d}$ is represented 
by its radial coordinate $r = R(\x) = \Vert \x \Vert$ 
and a $(d - 1)$-dimensional vector of angular coordinates 
$\tht = (\theta_1, \ldots, \theta_{d-1})^{\top}$ 
indicating the direction in space. 
More details will be given in the Appendix.

The radius $r = R(\x)$
is invariant under transformations from $SO(d)$,
i.\,e.\ $R(\g(\x)) = R(\x)$ for any rotation $\g \in SO(d)$. 
For a subsampling design $\xi$, we denote by
$\xi_{(R, \Th)}$ 
its representation (image) in terms of
hyperspherical coordinates
and by $\xi_{R}$ the marginal subsampling design
(projection) on the radius $r$.
The marginal subsampling design $\xi_{R}$ 
has total mass $\alpha$
and is bounded by the marginal distribution of
$R(X_{i}) = \Vert \X_{i} \Vert$, $f_{\xi_{R}}(r) \leq f_{R}(r)$.
Let $\bar{\mu}$ be the uniform (Haar) measure 
on the angle $\tht$ with total mass $1$
which is invariant with respect to
transformations from $SO(d)$
under consideration that the radius $R$ 
constitutes a maximal invariant 
\citep[see~e.\,g.\ ][]{wijsman1990invariant}.

For any subsampling design $\xi$, 
denote by $\bar{\xi}$ 
its symmetrization which has representation 
$\bar{\xi}_{(R, \Th)} = \xi_{R} \otimes \bar{\mu}$ 
in hyperspherical coordinates,
where ``$\otimes$'' is the common product of measures.
The symmetrization $\bar{\xi}$
is invariant with respect to transformations in $SO(d)$
(Lemma~\ref{lemma:decomposition})
and is itself a subsampling design
(Lemma~\ref{lemma:sym-design-bounded}).
The regression function $\f$
is linearly equivariant with respect to 
transformations in $SO(d)$ 
(see equation~\eqref{eq:multiple-regression-equivariance-sod}).
The $D$-criterion is convex and invariant with respect to $SO(d)$.
Then, according to
Theorem~\ref{theorem:optimality-of-sym-design},
any subsampling design $\xi$ is dominated
by its symmetrization $\bar{\xi}$,
\begin{equation*}
	\label{eq:domination-by-symmetrization}
	\det(\M(\xi)) \leq \det(\M(\bar{\xi})) \, .
\end{equation*}
Hence, we may restrict our search
for a $D$-optimal subsampling design
to the essentially complete class of invariant designs
$\bar{\xi}$ with representation
$\xi_{R} \otimes \bar{\mu}$.
In particular, we only have to optimize
the marginal subsampling design $\xi_{R}$
on the radius.

For any invariant subsampling design $\bar{\xi}$, 
all first order moments  
$\int x_{j} f_{\bar{\xi}} \diff \x$ 
and all mixed second order moments
$\int x_{j} x_{j^{\prime}} f_{\bar{\xi}} \diff \x$ 
of $\bar{\xi}$ are equal to zero,  
$j, j^{\prime} = 1, \ldots, d$, $j \neq j^{\prime}$,
by the representation 
$\xi_{R} \otimes \bar{\mu}$.
Further, all pure second order moments
$\int x_{j}^2 f_{\bar{\xi}} \diff \x$ 
are equal to $m_{2}(\bar{\xi}) > 0$, say.
The corresponding $(d + 1) \times (d + 1)$ 
information matrix $\M(\bar{\xi})$ is diagonal,
\[
	\M(\bar{\xi}) 
	= \begin{pmatrix}
			\alpha & \mathbf{0}
			\\
			\mathbf{0} & m_{2}(\bar{\xi}) \Id_{d}
	\end{pmatrix}
\]
(cf.\ Lemma~\ref{lemma:mean-of-rotated-designs}).

We can conclude from
Theorem~\ref*{theorem:opt-design-general-distribution}
that the $D$-optimal subsampling design $\xi_{\alpha}^{*}$ is concentrated
outside a $d$-dimensional sphere 
of appropriate size centered at $\mathbf{0}$.

\begin{theorem}
	\label{theorem:dens-of-opt-design-spherical}
	Let $d \geq 2$ and let the distribution of the covariates $\X_{i}$
	be centered and spherical.
	The subsampling design $\xi_{\alpha}^{*}$ is $D$-optimal
	if and only if
	$\xi_{\alpha}^{*}$ has density 
	\begin{equation}
		\label{eq:dens-of-opt-design-spherical}
		f_{\xi_{\alpha}^{*}}(\x) 
		= \left\{
			\begin{array}{cl}
				f_{\X}(\x) & \mbox{ for } \Vert \x \Vert^{2} \geq q_{1 - \alpha} \, ,
				\\
				0 & \mbox{ otherwise,}
			\end{array} 
		\right. 
	\end{equation}
	where $q_{1 - \alpha}$ is the $(1 - \alpha)$-quantile of the distribution of 
	$R(\X_{i})^{2} = \sum_{j=1}^{d} X_{ij}^{2}$.
\end{theorem}

Under the conditions of
Theorem~\ref{theorem:dens-of-opt-design-spherical},
the information matrix of the optimal subsampling design $\xi_{\alpha}^{*}$
has the shape
\begin{equation*}
	\label{eq:info-diagonal}
	\M(\xi_{\alpha}^{*}) 
	= \begin{pmatrix}
			\alpha & \mathbf{0}
			\\
			\mathbf{0} & m_{2}(\xi_{\alpha}^{*}) \Id_{d} \, .
		\end{pmatrix} \, .
\end{equation*}
The second moments 
$m_{2}(\xi_{\alpha}^{*}) =  \int x_{j}^2 f_{\xi_{\alpha}^{*}}(\x) \diff \x$
therein 
can be expressed in terms of the density $f_{R^{2}}$ of 
$R(\X_{i})^{2}$ as
\begin{equation}
	\label{eq:second-moment-spherical}
	m_{2}(\xi_{\alpha}^{*}) = \frac{1}{d} \int_{q_{1 - \alpha}}^{\infty} w f_{R^{2}}(w) \diff w \, .
\end{equation}
Obviously, 
$m_{2}(\xi_{\alpha}^{*}) > \alpha \sigma^{2}$ for all $\alpha \in (0,1)$.  

For $d = 1$, equation~\eqref{eq:dens-of-opt-design-spherical}
reduces to the condition for a $D$-optimal subsampling design 
in one covariate characterized in 
Corollary~\ref{corollary:opt-design-one-cov-0}.

\begin{example}[standard multivariate normal distribution]
	\label{example:standard-normal}
	In the case of a standard multivariate normal distribution
	of the covariates
	with mean $\mathbf{0}$ and covariance matrix $\Id_{d}$
	($ \X_{i}\sim \Nd_{d}\left(\mathbf{0}, \Id_{d} \right)$),
	the squared radius $R(\X_{i})^2$ is $\chi^2$-distributed
	with $d$ degrees of freedom. 
	Then, by Theorem~\ref{theorem:dens-of-opt-design-spherical}, 
	the $D$-optimal subsampling design $\xi_{\alpha}^{*}$
	includes all $\x$ outside the $d$-sphere
	with radius  $r^{*} = \sqrt{\chi^{2}_{d, 1 - \alpha}}$,
	where $\chi^{2}_{d, 1 - \alpha}$ is the $(1 - \alpha)$-quantile 
	of the $\chi^{2}$-distribution with $d$ degrees of freedom. 
	By the representation~\eqref{eq:second-moment-spherical},
	the second moments $m_{2}(\xi_{\alpha}^{*})$ 
	of the information matrix $\M(\xi_{\alpha}^{*})$
	can be calculated as
	\begin{equation}
		\label{eq:standard-normal-second-moment-optimal}
		m_{2}(\xi_{\alpha}^{*}) 
		= \alpha + \frac{2}{d} \chi^{2}_{d, 1 - \alpha} f_{\chi^2_d}(\chi^{2}_{d, 1 - \alpha}),
	\end{equation}
	where $f_{\chi^2_d}$ is the density of the $\chi^2$-distribution with $d$ degrees of freedom.
	In view of Corollary~\ref{corollary:opt-design-one-cov-0}, 
	we see that equation~\eqref{eq:standard-normal-second-moment-optimal} 
	also holds for $d = 1$.
	
	The second moment $m_{2}(\xi_{\alpha}^{*})$
	measures the percentage of information contained in
	the $D$-optimal subsampling design $\xi_{\alpha}^{*}$
	compared to the full data set,
	where the second moment is one,
	and to uniform random subsampling $\xi_{\mathrm{unif}}$,
	where the second moment is equal to $\alpha$.
	
	We plot these second moments 
	in Figure~\ref{Figure:second-moment}
	for various numbers $d$ of covariates
	in dependence on the subsampling proportion $\alpha$.
	As can be seen from the figure,
	all second moments are larger than $\alpha$
	in accordance with the remark following
	equation~\eqref{eq:second-moment-spherical}.
	For fixed dimension $d$,
	the second moment $m_{2}(\xi_{\alpha}^{*})$ decreases 
	when the sampling proportion $\alpha$ gets smaller
	which is obvious from the fact
	that the sample is getting smaller
	and, hence, estimation becomes less precise.
	In particular, $m_{2}(\xi_{\alpha}^{*})$ tends to $0$
	for $\alpha \to 0$. 
	For dimension $d = 1$,
	the second moment $m_{2}(\xi_{\alpha}^{*})$ 
	of $\xi_{\alpha}^{*}$ exceeds $\alpha$
	substantially for intermediate values of $\alpha$
	and, hence, the $D$-optimal subsampling design $\xi_{\alpha}^{*}$
	shows a substantially better performance than uniform random subsampling.
	This property is less pronounced for higher dimensions $d$.
	In particular, for fixed subsampling proportion $\alpha$,
	$m_{2}(\xi_{\alpha}^{*})$ decreases in the dimension $d$
	such that estimation becomes more difficult 
	when the dimension $d$ increases.
	For $d = 1\,000$,
	the second moment $m_{2}(\xi_{\alpha}^{*})$ is already rather close 
	to the value $\alpha$ for uniform random subsampling.
	We will discuss this behavior further 
	in terms of efficiency 
	in Example~\ref{example:normal-efficiency}
	below.
	\begin{figure}[h]
		\centering
		\includegraphics[width=.6\textwidth]{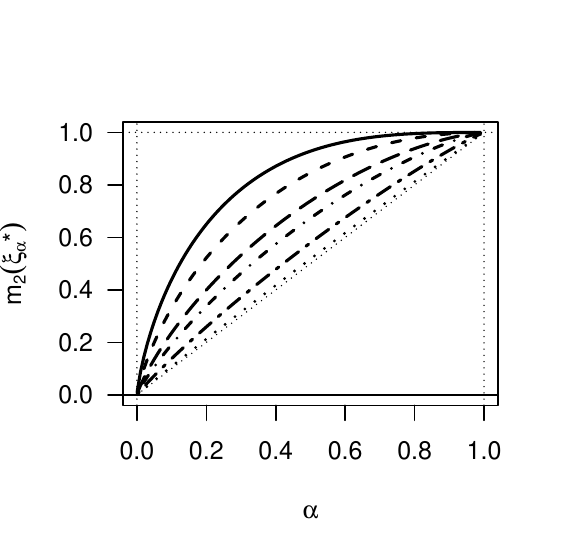}
		\caption{Second moment $m_{2}(\xi_{\alpha}^{*})$ 
			of the $D$-optimal subsampling design $\xi_{\alpha}^{*}$ 
			for standard (multivariate) normal distributions
			of dimensions $d = 1$ (solid), $2$ (dashes), $5$ (long dashes),
			$10$ (dashes and dots), $50$ (long and short dashes),
			and $1\,000$ (dots)
			in dependence on the subsampling proportion $\alpha$}
		\label{Figure:second-moment}
	\end{figure}
	
	To give an impression of the optimal subsampling design $\xi_{\alpha}^{*}$,
	we plot its marginal density $\xi_{R}^{*}$ on the radius
	in the case of standard bivariate normal covariates $\X_{i}$ 
	($d = 2$) 
	and subsampling proportion $\alpha = 0.1$ 
	in Figure~\ref{Figure:dens-of-radius}. 
	There, the solid line shows the density $\xi_{R}^{*}$
	on the radius for the subsampling design $\xi_{\alpha}^{*}$
	while the dashed line is
	the bounding density $f_{R(\X_{i})}$
	on the radius for the distribution of the covariates.
	The vertical line segment indicates the $(1 - \alpha)$-quantile 
	$\sqrt{\chi^{2}_{2, 0.9}} = 2.146$ of the marginal distribution 
	of the radius $R(\X_{i})$ of the covariates.

	\begin{figure}[h]
		\centering
		\includegraphics[width=.6\textwidth]{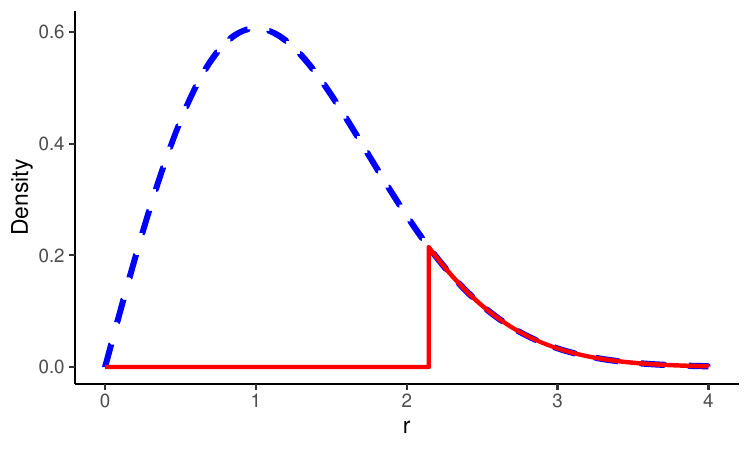}
		\caption{Density of the marginal optimal subsampling design $\xi_{R}^{*}$ 
			(solid) and the marginal distribution of the covariates $R(\X_{i})$ (dashed) 
			on the radius, 
			standard bivariate normal distribution, 
			subsampling proportion $\alpha = 0.1$}
		\label{Figure:dens-of-radius}
	\end{figure}
\end{example}
\begin{example}[multivariate $t$-distribution]
	\label{example:t-distribution}
	The distribution of $d$-dimensional $t$-distributed covariates $\X_{i}$
	with $\nu$ degrees of freedom
	may be defined by the ratio $\X_{i} = \Z_{i} / \sqrt{\mathbf{W}_{i} / {\nu}}$
	of a standard $d$-dimensional normal variate $\Z_{i}$
	and the square root of a standardized $\chi^2$ variate $\mathbf{W}_{i}$
	with $\nu$ degrees of freedom
	independent of each other.
	The covariates $\X_{i}$ are spherical and centered,
	and the standardized squared radius $R(\X_{i})^{2} / d$ 
	is $F$-distributed with $d$ and $\nu$ degrees of freedom.
	By Theorem~\ref{theorem:dens-of-opt-design-spherical}, 
	the $D$-optimal subsampling design $\xi_{\alpha}^{*}$
	includes all $\x$ outside the $d$-sphere
	with radius  $ r^{*} = \sqrt{F_{d, \nu, 1 - \alpha}} $,
	where $F_{d, \nu, 1 - \alpha}$ is the $(1 - \alpha)$-quantile 
	of the $F$-distribution with $d$ and $\nu$ degrees of freedom.
\end{example}

We will use the multivariate normal
and the multivariate $t$-distribution in Section~\ref{sec:method}
to examine the performance of 
subsampling procedures motivated by 
$D$-optimal subsampling designs. 

By equivariance considerations 
with respect to transformations of location and scatter
(see Lemma~\ref{lemma:opt-transformed-design}),
the result of Theorem~\ref{theorem:dens-of-opt-design-spherical}
can be extended to covariates $\X_{i}$ 
which have an elliptical distribution,
i.\,e.\ for which the density $f_{\X}$ has
elliptical contours such that 
$f_{\X}(\x) = f_0\left((\x - \mx)^{\top} \bm{\Sigma}^{-1} (\x - \mx)\right)$ 
for some univariate function $f_0$, 
location vector $\mx$, and
positive-definite dispersion matrix $\bm{\Sigma}$.
Note that
$\mx = \E[\X_{i}]$, and $\bm{\Sigma}$ can be chosen as
$\Cov[\X_{i}]$ so that
$(\x - \mx)^{\top} \bm{\Sigma}^{-1} (\x - \mx)$
is the Mahalanobis distance 
$\mathrm{d}_{\bm{\Sigma}}(\x, \mx)$ 
of $\x$ and $\mx$ with respect to $\bm{\Sigma}$.

\begin{theorem}
	\label{theorem:dens-of-opt-design-elliptical}
	Let $d \geq 2$ and let the distribution of the covariates $\X_{i}$
	be elliptical with mean $\mx$ 
	and covariance matrix $\bm{\Sigma}$.
	The subsampling design $\xi_{\alpha}^{*}$ is $D$-optimal
	if and only if
	$\xi_{\alpha}^{*}$ has density 
	\begin{equation*}
		\label{eq:dens-of-opt-design-elliptical}
		f_{\xi_{\alpha}^{*}}(\x) 
		= \left\{
		\begin{array}{cl}
			f_{\X}(\x) & \mbox{ for } (\x - \mx)^{\top} \bm{\Sigma}^{-1} (\x - \mx) \geq q_{1 - \alpha} \, ,
			\\
			0 & \mbox{ otherwise,}
		\end{array} 
		\right. 
	\end{equation*}
	where $q_{1 - \alpha}$ is the $(1 - \alpha)$-quantile of the distribution of 
	$(\X_{i} - \mx)^{\top} \bm{\Sigma}^{-1} (\X_{i} - \mx)$.
\end{theorem}

The $D$-optimal subsampling design is, hence, concentrated 
on the complement of the interior of the concentration ellipsoid 
which contains mass $1 - \alpha$ of the distribution of $\X_{i}$.
Moreover, for elliptical distributions,
the optimality conditions 
in Theorem ~\ref{theorem:opt-design-general-distribution}
and Theorem~\ref{theorem:dens-of-opt-design-elliptical}
coincide 
whereat $\m(\xi_{\alpha}^{*}) = \mx$, 
$\mathbf{S}(\xi_{\alpha}^{*}) = s^{2}(\xi_{\alpha}^{*}) \bm{\Sigma}$,
$c = q_{1 - \alpha} / s^{2}(\xi_{\alpha}^{*})$,
and 
$s^{2}(\xi)
	= \frac{1}{d} \int (\x - \mx)^{\top} \bm{\Sigma}^{-1} (\x - \mx) f_{\xi}(\x) \diff \x$ 
is the scaled (per covariate) average Mahalanobis distance
under the subsampling design $\xi$.

\begin{example}[general multivariate normal distribution]
	\label{example:normal}
	We extend our findings from Example~\ref{example:standard-normal} 
	for the standard multivariate normal distribution of the covariates
	to the situation of a general multivariate normal distribution 
	$\X_{i} \sim \Nd_{d}\left(\mx, \bm{\Sigma} \right)$
	with mean $\mx$ and covariance matrix $\bm{\Sigma}$. 
	By Theorem~\ref{theorem:dens-of-opt-design-elliptical}, 
	the $D$-optimal subsampling design $\xi_{\alpha}^{*}$ is equal 
	to the distribution of the $\X_{i}$ 
	on the complement 
	$(\x - \mx)^{\top} \bm{\Sigma}^{-1} (\x - \mx) \geq \chi^{2}_{d, 1 - \alpha}$
	of the $(1 - \alpha)$ concentration ellipsoid. 
\end{example}

In the literature, prevalent interest is often
in estimating the slope parameters 
$\B_{\mathrm{slope}} = (\beta_{1}, \ldots, \beta_{d})^{\top}$ 
disregarding the intercept $\beta_{0}$ 
\citep[see e.\,g.][]{wang2019information}.
Then the quality of a subsample
is measured in terms 
of the asymptotic covariance matrix 
of the vector
$\hat{\B}_{\mathrm{slope}} = (\hat\beta_{1}, \ldots, \hat\beta_{d})^{\top}$ 
of slope parameter estimators.
For a subsampling design $\xi$
the asymptotic covariance matrix 
of $ \hat{\B}_{\mathrm{slope}}$
is proportional to the lower right $d \times d$ submatrix 
$\mathbf{S}(\xi)^{-1}$
of the inverse $\M(\xi)^{-1}$ of the information matrix $\M(\xi)$,
where 
$\mathbf{S}(\xi)$
is defined as in equation~\eqref{eq:concentration ellipsoid-mean-scale}.
The determinant 
$\det(\mathbf{S}(\xi))$ 
for the slopes $\B_{\mathrm{slope}}$ 
and the determinant
$\det(\M(\xi)) = \alpha \det(\mathbf{S}(\xi))$
for the full parameter vector $\B$ 
differ only by the constant factor $\alpha$.
Hence, the $D$-optimal subsampling design $\xi_{\alpha}^{*}$ 
for the full parameter vector $\B$ 
is also $D_{\mathrm{slope}}$-optimal
for the slope vector $\B_{\mathrm{slope}}$.

For the $D$-optimal subsampling design $\xi_{\alpha}^{*}$,
the slope estimator $\hat{\B}_{\mathrm{slope}}$
is asymptotically normal 
with asymptotic covariance matrix
\begin{equation}
	\label{eq:as-var-slope}
	\mathrm{as.}\!
	\Cov(\hat{\B}_{\mathrm{slope}})	
	= \frac{\sigma_{\e}^{2}}{s^{2}(\xi_{\alpha}^{*})} \bm{\Sigma}^{-1} .
\end{equation}

In particular,
when the distribution
of the covariates is spherical
with mean $\mx = (\mu_{1}, \ldots, \mu_{d})^{\top}$,
the slope related information matrix 
of the $D_{\mathrm{slope}}$-optimal subsampling design $\xi_{\alpha}^{*}$
is equal to 
$\mathbf{S}(\xi_{\alpha}^{*}) = s^{2}(\xi_{\alpha}^{*}) \Id_{d}$.
Then the asymptotic variance of $\hat{\beta}_{j}$
is $1 / s^{2}(\xi_{\alpha}^{*})$ for any component 
$\beta_{j}$ of the slope vector $\B_{\mathrm{slope}}$.
The quantity $s^{2}(\xi_{\alpha}^{*})$ 
may be interpreted as the marginal dispersion
$\int (x_{j} - \mu_{j})^{2} f_{\xi_{\alpha}^{*}}(\x) \diff \x$ 
of $\xi_{\alpha}^{*}$
in any direction $x_{j}$.
If, moreover, the distribution of the covariates is centered,
then the dispersion $s^{2}(\xi_{\alpha}^{*})$ 
is equal to the second moment $m_{2}(\xi_{\alpha}^{*})$
of $\xi_{\alpha}^{*}$ and, hence,
$\mathbf{S}(\xi_{\alpha}^{*}) = m_{2}(\xi_{\alpha}^{*}) \Id_{d}$.

For later use in Section~\ref{sec:method},
we add the following property of the dispersion measure 
$s^{2}(\xi_{\alpha}^{*})$.

\begin{lemma}
	\label{lemma:unbounded-distribution}
	If the distribution of the covariates 
	is elliptical and unbounded,
	then
	$\lim_{\alpha \to 0} s^{2}(\xi_{\alpha}^{*}) / \alpha = \infty$. 
\end{lemma}

Note that
$s^{2}(\xi_{\alpha}^{*}) / \alpha$ remains bounded
when the covariates have a bounded distribution.

For measuring the quality of a subsampling design $\xi$
with subsampling proportion $\alpha$,
we make use of the $D_{\mathrm{slope}}$-efficiency 
\begin{equation}
	\label{eq:d-slope-efficiency}
	\eff_{D, \mathrm{slope}}(\xi) 
	= \left(\frac{\det(\mathbf{S}(\xi))}{\det(\mathbf{S}(\xi_{\alpha}^{*}))}\right)^{1/d} .
\end{equation}
Here, we employ the homogeneous version $\det(\mathbf{S}(\xi))^{1/d}$
of the $D_{\mathrm{slope}}$-criterion 
satisfying the homogeneity condition
$\det(\lambda \mathbf{S})^{1/d} = \lambda \det(\mathbf{S})^{1/d}$
for any $\lambda > 0$
\cite[see][Chapter~6.2]{pukelsheim1993optimal}.
The efficiency $\eff_{D, \mathrm{slope}}(\xi)$ 
might be interpreted straightforwardly 
in terms of the size $n$ of the full data set 
and, hence, of the size $\alpha n$ of the subsample:
When the subsampling design $\xi$ is used,
a full data set of size 
$n^{\prime} = n / \eff_{D, \mathrm{slope}}(\xi) \geq n$  
would be required to obtain the same value 
of the $D_{\mathrm{slope}}$-criterion 
as when the $D_{\mathrm{slope}}$-optimal 
subsampling design $\xi_{\alpha}^*$ would have been used
on a full data set of size $n$.
Accordingly, also the size $n^{\prime} \alpha \geq n \alpha$
of the subsample has to be increased 
when $\xi$ is used to maintain the precision
of the $D_{\mathrm{slope}}$-optimal 
subsampling design $\xi_{\alpha}^*$.
But the size $n$ of the full data set 
is typically not at the disposition of the examiner.

Nevertheless,
when we consider uniform random subsampling 
$\xi_{\mathrm{unif}}$ 
with density 
$f_{\xi_{\mathrm{unif}}}(\x) = \alpha f_{\X} (\x)$
for subsampling proportion $\alpha$
as a natural choice 
with which to compare 
the optimal subsampling design $\xi_{\alpha}^{*}$,
the efficiency $\eff_{D, \mathrm{slope}}(\xi_{\mathrm{unif}})$ 
can be nicely interpreted in terms of the subsampling proportion 
as has been pointed out in \cite{reuter2023optimal}: 
For a full data set of fixed size $n$, 
a uniform random subsampling design
with subsampling proportion 
$\alpha^{\prime} = \alpha / \eff_{D, \mathrm{slope}}(\xi_{\mathrm{unif}}) \geq \alpha$  
would be required 
to obtain the same precision 
in terms of the $D_{\mathrm{slope}}$-criterion 
as when the $D_{\mathrm{slope}}$-optimal 
subsampling design $\xi_{\alpha}^*$ 
of subsampling proportion $\alpha$
would have been used.
For example, if the efficiency $\eff_{D, \mathrm{slope}}(\xi_{\mathrm{unif}})$ is $0.5$, 
then twice as many units would be needed 
in the subsample
under uniform random subsampling 
than for the $D_{\mathrm{slope}}$-optimal subsampling design
to obtain the same precision 
in terms of the $D_{\mathrm{slope}}$-criterion.

In the case of a spherical centered distribution
of the covariates, the 
$D_{\mathrm{slope}}$-efficiency~\eqref{eq:d-slope-efficiency}
of uniform random subsampling reduces to
\begin{equation}
	\label{eq:d-slope-efficiency-spherical}
	\eff_{D, \mathrm{slope}}(\xi_{\mathrm{unif}}) 
	= \frac{\alpha \sigma^{2}}{m_{2}(\xi_{\alpha}^{*})} .
\end{equation}

By considerations of equivariance,
the $D_{\mathrm{slope}}$-efficiency
of uniform random subsampling is invariant  with respect to 
affine linear transformations of the covariates. 

\begin{example}[multivariate normal distribution]
	\label{example:normal-efficiency}
	When the covariates are multivariate normal 
	($ \X_{i} \sim \Nd_{d}\left(\mx, \bm{\Sigma} \right)$),
	the efficiency of uniform random subsampling is
	\begin{equation*}
		\label{eq:standard-normal-efficiency-uniform-random}
		\eff_{D, \mathrm{slope}}(\xi_{\mathrm{unif}}) 
		= \frac{d \alpha}{d \alpha + 2 \chi^{2}_{d, 1 - \alpha} f_{\chi^2_d}(\chi^{2}_{d, 1 - \alpha})},
	\end{equation*}
	by equations~\eqref{eq:standard-normal-second-moment-optimal}
	and~\eqref{eq:d-slope-efficiency-spherical}.
	In Figure~\ref{Figure:efficiency-uniform-normal},
	we plot the $D_{\mathrm{slope}}$-efficiency
	of uniform random subsampling
	for various numbers $d$ of covariates
	in dependence on the subsampling proportion $\alpha$.
	As can be seen from the figure,
	the $D_{\mathrm{slope}}$-efficiency
	of uniform random sampling is always larger than $\alpha$.
	This is in accordance with the argument 
	in~\cite{reuter2023optimal} 
	that uniform random subsampling 
	has relative efficiency $\alpha$
	compared to the full data set
	and the optimal subsampling design $\xi_{\alpha}^{*}$
	bears less information than full data.
	For fixed dimension $d$,
	the $D_{\mathrm{slope}}$-efficiency
	of uniform random sampling decreases 
	when the sampling proportion $\alpha$ gets smaller,
	and approaches zero for $\alpha \to 0$.
	For dimension $d = 1$,
	the $D_{\mathrm{slope}}$-efficiency
	of uniform random sampling is close to $\alpha$.
	This property is less pronounced for higher dimensions $d$.
	In particular, for fixed subsampling proportion $\alpha$,
	the $D_{\mathrm{slope}}$-efficiency
	of uniform random sampling increases in the dimension $d$
	and tends to $1$ for $d \to \infty$.
	For $d = 1\,000$,
	the $D_{\mathrm{slope}}$-efficiency
	of uniform random sampling is already quite high
	for reasonable values of the subsampling proportion $\alpha$
	($\eff_{D, \mathrm{slope}}(\xi_{\mathrm{unif}}) \geq 0.89$
	when $\alpha \geq 0.01$).
	Thus, there is a substantial gain in using the 
	$D$-optimal subsampling design $\xi_{\alpha}^{*}$
	instead of uniform random subsampling
	in the case of small to moderate dimension $d$.
	But the gain is less prominent for higher dimensions $d$.
	\begin{figure}[h]
		\centering
		\includegraphics[width=.6\textwidth]{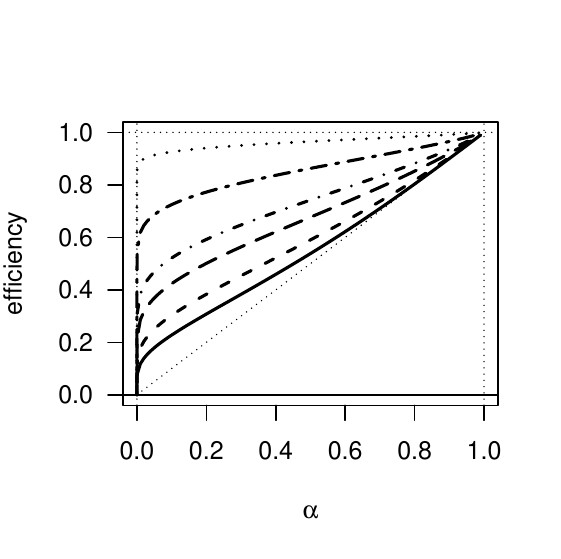}
		\caption{Efficiency of uniform random subsampling
			for multivariate normal distributions
			of dimensions $d = 1$ (solid), $2$ (dashes), $5$ (long dashes),
			$10$ (dashes and dots), $50$ (long and short dashes),
			and $1\,000$ (dots)
			in dependence on the subsampling proportion $\alpha$}
		\label{Figure:efficiency-uniform-normal}
	\end{figure}
\end{example}

\section{Subsampling Algorithms}
\label{sec:algorithms}

To implement a feasible subsampling procedure 
according to the $D$-optimal subsampling design 
$\xi_{\alpha}^{*}$ from Theorem~\ref{theorem:dens-of-opt-design-elliptical}, 
we first propose the following Algorithm~\ref{alg:quant}. 

\begin{algorithm}
	\caption{Subsample selection according to 
		the $D$-optimal subsampling design $\xi_{\alpha}^{*}$}
	\label{alg:quant}
	\KwData{Covariates $\x_{i}$, $i = 1, \ldots, n$, mean $\mx$, 
		covariance matrix $\bm{\Sigma}$.}
	Fix $\alpha$\;
	For $i = 1, \ldots, n$ do:
	\\
	Step 1: Calculate the Mahalanobis distance 
	$\mathrm{d}_{\bm{\Sigma}}(\x_{i}, \mx) = (\x_{i} - \mx)^{\top} \bm{\Sigma}^{-1} (\x_{i} - \mx)$\;
	Step 2: Select $ \x_{i} $ when $\mathrm{d}_{\bm{\Sigma}}(\x_{i}, \mx) \geq  q_{1 - \alpha}$\;
	Repeat\;
\end{algorithm}

Algorithm~\ref{alg:quant}
provides a simple acceptance-rejection method 
in which all data points are accepted into the subdata 
that lie in the support of $ \xi_{\alpha}^{*} $
while all other data points are rejected. 
This selection procedure can be performed
sequentially by looking at each data point once
and decide instantly on acceptance, 
irrespectively of all other data 
\citep[cf.][]{pronzato2021sequential}.

As the covariates $\X_{i}$ are random, the 
Algorithm~\ref{alg:quant}
results in a random size $K$ of 
the subsample $(\X^{\prime}_{1}, \ldots, \X^{\prime}_{K})$, say.
The subsample size $K$ is binomial, 
$K \sim \mathcal{B}(n, \alpha)$
with size $n$ of the full data set and 
subsampling proportion $\alpha$.
To assess the performance of the algorithm,
we consider the asymptotic behavior
when the size $n$ of the full data set 
and, hence, the subsample size $K = K_{n}$
go to infinity.
Then, by the Law of Large Numbers, the proportion
$K_{n} / n$ of data selected tends to $\alpha$.
The elements $\X^{\prime}_{i}$ of the subsample
are independent with density
$f_{\X^{\prime}_{i}}(\x) = \alpha^{-1} f_{\xi_{\alpha}^{*}}(\x)$,
and the standardized information
$n^{-1} \sum_{i=1}^{K_{n}} \f(\X^{\prime}_{i}) \f(\X^{\prime}_{i})^{\top}$
tends to $\M(\xi_{\alpha}^{*})$.
Moreover, the associated least squares estimator
$\hat{\B}_{n}$ is asymptotically normal
with asymptotic covariance matrix 
$\sigma_{\e}^{2} \M(\xi_{\alpha}^{*})^{-1}$
(see Lemma~\ref{lemma:asymptotic-normality}).

To achieve a deterministic subsample size $k$, say,
with subsampling proportion $k / n \approx \alpha$, 
one may adopt a strategy presented
in~\cite{pronzato2006construction}.
However, we propose the following, 
simpler nonsequential Algorithm~\ref{alg:topk}.
To state this algorithm, we introduce the notation $\x_{i:n}$ 
for the $i$th generalized (reverse) order statistics 
based on the Mahalanobis distance 
$\mathrm{d}_{\bm{\Sigma}}(\x, \mx) = (\x - \mx)^{\top} \bm{\Sigma}^{-1} (\x - \mx)$ 
such that $(\x_{1:n}, \ldots, \x_{n:n})$ 
is a permutation of $(\x_{1}, \ldots, \x_{n})$ and 
$\mathrm{d}_{\bm{\Sigma}}(\x_{i:n}, \mx) 
	\geq \mathrm{d}_{\bm{\Sigma}}(\x_{i+1:n}, \mx)$.
The latter inequalities are strict almost surely
by the continuity of the distribution of the covariates $\X_{i}$.

\begin{algorithm}
	\caption{Subsample selection according to maximal Mahalanobis distance}
	\label{alg:topk}
	\KwData{Covariates $\x_{i}$, $i = 1, \ldots, n$, mean $\mx$, 
		covariance matrix $\bm{\Sigma}$.}
	Fix $k$\;
	Step 1: For $i = 1, \ldots, n$ do: 
	\\
	\phantom{Step 1: }
	Calculate the Mahalanobis distance 
	$\mathrm{d}_{\bm{\Sigma}}(\x_{i}, \mx) = (\x_{i} - \mx)^{\top} \bm{\Sigma}^{-1} (\x_{i} - \mx)$\;
	\phantom{Step 1: }
	Repeat\;
	Step 2: Select $\x_{1:n}, \ldots, \x_{k:n}$ 
	corresponding to the $k$ largest values of $\mathrm{d}_{\bm{\Sigma}}(\x_{i}, \mx)$\;
\end{algorithm}

The selection Step~2 of Algorithm~\ref{alg:topk} 
can be done e.\,g.~by using partial quicksort 
\citep[see][]{martinez2004partial}.
Algorithm~\ref{alg:topk}
has the additional advantage
that it does not rely on the particular distribution 
of the covariates apart from ellipticity
and does not need calculation of any
quantile.  
Only, knowledge of the first and second moments 
is requested which may be estimated from the data.

To obtain a subsample with subsampling proportion 
approximately $\alpha$,
the subsample size $k$ may be chosen
as the integer part $k = k_{n} = [n \alpha]$ 
of $n \alpha$.
When the size $n$ of the full data increases,
the Mahalanobis distance
$\mathrm{d}_{\bm{\Sigma}}(\X_{k_{n}:n}, \mx)$ 
of the $k_{n}$th order statistics $\X_{k_{n}:n}$
tends to the $(1 - \alpha)$-quantile $q_{1 - \alpha}$,
and the asymptotic properties 
of the subsample obtained by Algorithm~\ref{alg:topk} 
are similar to those
of the subsample generated by Algorithm~\ref{alg:quant}. 
Thus, the inverse information matrix $\M(\xi_{\alpha}^{*})^{-1}$
may serve as an approximation to the asymptotic covariance
of the least squares estimator $\hat{\B}$
based on the observations in the subsample 
$\X_{1:n}, \ldots, \X_{k_{n}:n}$,
\begin{equation*}
	\label{eq:variance-approx}
	\Cov[\hat{\B}_{n} ; \X_{1:n}, \ldots, \X_{k_{n}:n}] 
	\approx \frac{1}{n} \sigma_{\e}^{2} \M(\xi_{\alpha}^{*})^{-1} . 
\end{equation*}
This approach will be supported by the simulation results below.

\section{Subsampling Design with Fixed Sample Size, Simulation}
\label{sec:method}

In contrast to the previous sections, 
where we aim at subsampling 
a certain proportion $\alpha$ of the full data, 
we now consider the case of selecting 
a fixed number $k$ of data points
as in \cite{wang2019information}
while the size $n$ of the full data may vary.
In this situation, the subsampling proportion $\alpha_{n} = k / n$ decreases 
when $ n $ increases.
Although there will be no straightforward 
asymptotic behavior in $n$ for $k$ fixed,
we propose to use the approximation 
by continuous subsampling designs
$\xi_{n}$ with total mass $\alpha_{n} = k / n $
as in Section~\ref{sec:design}
if the subsampling size $k$ is sufficiently large.

To allow for comparison of different sizes $n$
of the full data set, 
we will use the nonstandardized (per subsample) information matrix
$\M_{n}(\xi_{n}) = n \int \f(\x) \f(\x)^{\top} f_{\xi_{n}}(\x) \diff \x$ 
from now on
such that $n \int f_{\xi_{n}}(\x) \diff \x = k$
for fixed subsampling size $k$.
When $k$ is large, 
the asymptotic results of the previous sections 
give rise to consider the inverse information matrix 
$\M_n(\xi_{n})^{-1}$ as an approximation to the covariance matrix 
of the least squares estimator $\hat{\B}_{n}$ 
based on the subsample of $k$ out of $n$ data points
according to $\xi_{n}$. 
Hence, it is reasonable to make use 
of the optimal continuous subsampling design $\xi_{\alpha_{n}}^{*}$ 
for the proportion $\alpha_{n} = k / n$
as derived in 
Theorem~\ref{theorem:dens-of-opt-design-elliptical}.

In the following simulation study,
we will generate subsamples 
by Algorithm~\ref{alg:topk} 
with $k$ fixed for various values of $n$ for the full data size.
We obtain subsamples 
$\X_{1:n}, \ldots, \X_{k:n}$ 
which
consists of those $k$ data points with
largest Mahalanobis distance 
$\mathrm{d}_{\bm{\Sigma}}(\X_{i}, \mx)$
from the mean $\mx$.
Conditionally on 
$ \X_{1:n}, \ldots, \X_{k:n} $, 
these subsamples 
have \textit{observed} nonstandardized information matrix
$\M(\X_{1:n}, \ldots, \X_{k:n}) = \sum_{i=1}^{k} \f(\X_{i:n}) \f(\X_{i:n})^{\top}$,
and the mean information
$\E[\M(\X_{1:n}, \ldots, \X_{k:n})]$
may be approximated by
$\M_{n}(\xi_{k/n}^{*}) = n \M(\xi_{k/n}^{*})$.

Accordingly,
when we are interested in the slopes only,
the observed slope related information matrix
$\mathbf{S}(\X_{1:n}, \ldots, \X_{k:n})$
is the inverse of the lower right
$d \times d$ submatrix
of $\M(\X_{1:n}, \ldots, \X_{k:n})^{-1}$,
and its mean 
may be approximated by
$n \mathbf{S}(\xi_{k/n}^{*}) = n s^{2}(\xi_{k/n}^{*}) \bm{\Sigma}$.

Similar to other simulation studies in the literature,
we will consider the variances
of the slope estimates $\hat{\B}_{\mathrm{slope}}$.
The covariance matrix of $\hat{\B}_{\mathrm{slope}}$
may be decomposed,
\begin{equation}
	\label{eq:covariance-topk-decompose}
	\Cov[\hat{\B}_{\mathrm{slope}}]
	= \E\left[\Cov[\hat{\B}_{\mathrm{slope}} | \X_{1:n}, \ldots, \X_{k:n}]\right]
		+ \Cov\left[\E[\hat{\B}_{\mathrm{slope}} | \X_{1:n}, \ldots, \X_{k:n}]\right] \, ,
\end{equation}
into the expectation of the conditional covariance
and the covariance of the conditional expectation 
given the covariates, respectively.
Since the slope estimator $\hat{\B}_{\mathrm{slope}}$
is conditionally unbiased, 
the latter term in equation~\eqref{eq:covariance-topk-decompose}
vanishes,
and the conditional covariance 
$\Cov[\hat{\B}_{\mathrm{slope}} | \X_{1:n}, \ldots, \X_{k:n}]$
is proportional to the inverse 
of the slope related information
$\mathbf{S}(\X_{1:n}, \ldots, \X_{k:n})$.
Hence,
\begin{equation*}
	\label{eq:covariance-topk}
	\Cov[\hat{\B}_{\mathrm{slope}}]
		= \sigma_{\e}^{2} \E\left[\mathbf{S}(\X_{1:n}, \ldots, \X_{k:n})^{-1}\right] \, .
\end{equation*}
For $k$ large,
the covariance $\Cov[\hat{\B}_{\mathrm{slope}}]$ 
may be approximated by
its asymptotic counterpart~\eqref{eq:as-var-slope},
\begin{equation}
	\label{eq:var-of-slope-fixed-k}
	\Cov[\hat{\B}_{\mathrm{slope}}]
	\approx \frac{\sigma_{\e}^{2}}{n s^{2}(\xi_{k/n}^{*})} \bm{\Sigma}^{-1} .
\end{equation}
Note that, 
by Lemma~\ref{lemma:unbounded-distribution},
the leading term on the right hand side 
of equation~\eqref{eq:var-of-slope-fixed-k}
will tend to zero for $n$ to infinity
when the distribution of the covariates is unbounded.
This indicates a kind of consistency of $\hat{\B}_{\mathrm{slope}}$
in increasing size $n$ of the full data set
although the sample size $k$ remains fixed
as has been observed in
\cite{wang2019information}
for their subsampling method IBOSS
and will be supported by our simulations below.

\begin{example}[standard multivariate normal distribution]
	\label{example:standard-normal-mse-approximate}
	In the case of standard multivariate normally distributed covariates, 
	$\X_{i} \sim \Nd_{d}(\mathbf{0}, \Id_{d})$,
	we get the approximation 
	\begin{equation}
		\label{eq:var-of-slope-normal-dist}
		\Cov[\hat{\B}_{\mathrm{slope}}]
		\approx 
		\frac{1}{n s^{2}(\xi_{k/n}^{*})} \Id_{d}
		=
		\left( k 
			+ \frac{2 n}{d} \chi_{d, 1 - (k/n)}^{2} 
					f_{\chi_{d}^{2}}\left(\chi_{d, 1 - (k/n)}^{2}\right) 
		\right)^{-1} \Id_{d}.
	\end{equation}	
	by equations~\eqref{eq:var-of-slope-fixed-k}
	and \eqref{eq:standard-normal-second-moment-optimal}.
	The mean squared error 
	$ \MSE(\hat{\B}_{\mathrm{slope}}) = \sum_{j=1}^{d} \Var[\hat{\beta}_{j}]$
	considered in \cite{wang2019information}
	is the trace of $\Cov[\hat{\B}_{\mathrm{slope}}]$.
	In order to compare the behavior 
	for varying dimensions $d$, 
	we use the standardized (per dimension) 
	mean squared error,
	$\MSE(\hat{\B}_{\mathrm{slope}}) / d$ 
	which is, in the present situation, 
	equal to $\Var[\hat{\beta}_{j}]$ for estimating the slope 
	$\beta_{j}$ of any component of the covariates. 
	In Figure~\ref{Figure:Variances_ND}, the plotted lines 
	depict the approximation 
	$d / \left(d k + 2 n \chi_{d, 1 - (k/n)}^{2} f_{\chi_{d}^{2}}(\chi_{d, 1 - (k/n)}^{2})\right)$
	of $\MSE(\hat{\B}_{\mathrm{slope}}) / d$
	from equation~\eqref{eq:var-of-slope-normal-dist} 
	for $d = 2$, $5$, $10$, $25$, and $50$
	in dependence on the size $n$ of the full data 
	while the size $k = 1\,000$ of the subsample is fixed. 
	Values of the (approximated) standardized $\MSE$
	are indicated by the labels on the left vertical axis. 
	The results are in accordance with 
	Example~\ref{example:standard-normal}:
	For any dimension $d$,
	the $\MSE$ decreases in the full data size $n$
	and tends to $0$ for $n \to \infty$.
	This behavior is less pronounced 
	for larger dimension $d$
	because estimation becomes more difficult
	when the number of parameters increases.
	\begin{figure}[h]
		\centering
		\includegraphics[width=.6\textwidth]{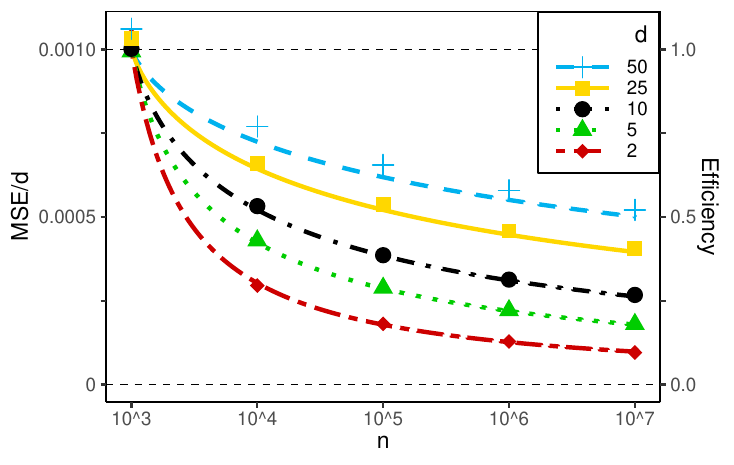}
		\caption{Approximate (lines) and simulated (symbols) 
			standardized mean squared errors 
			and approximate efficiency of uniform random subsampling 
			in dependence on full data size $n$, 
			subsample size $k = 1\,000$, 
			and various numbers $d$ of standard normal covariates}
		\label{Figure:Variances_ND}
	\end{figure}
	
	Additionally,
	in Figure~\ref{Figure:Variances_ND},
	the symbols represent corresponding simulated values
	of $\MSE / d$ for selected numbers 
	$n = 10^{k}$, $k = 3, \ldots, 7$,
	for the size of the full data set. 
	
	For this simulation, we generate complete full data sets
	and compute the simulated mean squared error as follows:
	In each iteration $v = 1, \ldots, V = 1\,000$, 
	\begin{itemize}
		\item 
			the parameter vector $\B^{(v)}$ is generated from 
			a standard multivariate normal distribution of dimension $d + 1$,
			$\B^{(v)} \sim \Nd_{d+1}(\mathbf{0}, \Id_{d+1})$,
		\item 
			the covariates $\x_{i}^{(v)}$ come from a 
			$d$-dimensional standard multivariate normal distribution,
			$\X_{i}^{(v)} \sim \Nd_{d}(\mathbf{0}, \Id_{d})$,
		\item 
			the error terms $\e_{i}^{(v)}$ come from
			a standard normal distribution,
			$\e_{i}^{(v)} \sim \Nd(0, 1)$,
		\item 
			and the values $y_{i}^{(v)}$ of the response variable
			are obtained by 
			$y_{i}^{(v)} 
				= \beta_{0}^{(v)} + {\x_{i}^{(v)}}^{\top} \B_{\mathrm{slope}}^{(v)} + \e_{i}^{(v)}$. 
		\item 
			For each size $n$, 
			we select subdata according to Algorithm~\ref{alg:topk}
			and compute the least squares estimate $ \hat{\B}_{n}^{(v)}$. 
		\item
			From these estimates, we calculate the simulated mean squared error
			\\
			$\MSE(\hat{\B}_{\mathrm{slope}}) 
				= \frac{1}{V} \sum_{v=1}^{V} \Vert \hat{\B}_{\mathrm{slope}}^{(v)} - \B^{(v)} \Vert^{2}$.	
	\end{itemize}
	
	From Figure~\ref{Figure:Variances_ND}
	we see that the simulated standardized mean squared error
	$\MSE(\hat{\B}_{\mathrm{slope}}) / d$
	tends to zero as $n$ goes to infinity. 
	While this decrease is evident for low dimensions $d$,
	it turns out to be substantially slower
	for higher dimensions 
	as more parameters need to be estimated
	from the same number $k$ of observations.
	It can be seen that the 
	approximated $\MSE$ values are close 
	to the simulated ones, at least, for small to moderate dimensions $d$.
	This justifies the approximation proposed in 
	equation~\eqref{eq:var-of-slope-normal-dist}.
	However, the simulated $\MSE$ 
	is systematically larger than 
	the approximate $\MSE$.
	This observation may be explained by noticing
	that the simulated covariance matrix estimates
	$\E[\M(\X_{1:n}, \ldots, \X_{k:n})^{-1}]$
	which is larger than the approximate covariance matrix 
	$\E[\M(\X_{1:n}, \ldots, \X_{k:n})]^{-1}$ 
	by Jensen's inequality.
	The exceedance is more pronounced
	for higher dimensions $d$.
	
	The relative efficiency
	of uniform random subsampling
	can be defined 
	in terms of $\MSE$ 
	as the ratio of the $\MSE$ under 
	$\X_{1:n}, \ldots, \X_{k:n}$
	divided by the $\MSE$ under
	uniform random subsampling.
	This ratio can be approximated
	by $k / \left(n s^{2}(\xi_{k/n}^{*})\right)$
	(see~\eqref{eq:d-slope-efficiency-spherical}).
	Hence,
	the efficiency of the
	uniform random subsampling design
	is $k$ times the approximation
	of the standardized $\MSE$ in equation~\eqref{eq:var-of-slope-normal-dist}.
	As a consequence, Figure~\ref{Figure:Variances_ND}
	also depicts the relative efficiency
	of uniform random subsampling,
	when the right vertical axis is used.
\end{example}

The $\MSE$ considered in
Example~\ref{example:standard-normal-mse-approximate}
corresponds to the $A$-criterion 
for estimating the slope parameters
in classical optimal design theory.
Hence, for spherical distributions of the covariates,
the $D$-optimal subsampling design $\xi_{\alpha}^{*}$
is also $A$-optimal for $\MSE(\hat{\B}_{\mathrm{slope}})$.
Then, under $\xi_{\alpha}^{*}$, 
the approximate standardized mean squared error 
$\MSE / d$ for the slopes
coincides with the inverse homogeneous version
$\det(\mathbf{S}(\xi_{\alpha}^{*}))^{-1/d}$
of the $D_{\mathrm{slope}}$-criterion.
However, in contrast to the $D$-criterion,
the $\MSE$-criterion is not equivariant
with respect to linear transformations,
and the $D$-optimal subsampling design $\xi_{\alpha}^{*}$
does not remain to be optimal 
with respect to the $\MSE$
when the elliptical distributions of the covariates
is nonspherical. 
For our proposed subsampling scheme
$\X_{1:n}, \ldots, \X_{k:n}$ of Algorithm~\ref{alg:topk},
we will thus consider the $D_{\mathrm{slope}}$-criterion
$\det(\mathbf{S}(\xi_{\alpha}^{*}))^{-1/d}$
instead of $\MSE / d$
in the subsequent simulation studies.

Further, note that, in the simulation of 
Example~\ref{example:standard-normal-mse-approximate},
the simulated values of the parameter vector $\B$ 
do not have any influence on the estimated variances
and, hence on the simulation results. 
Therefore, there is no need to generate $\B$
in the simulation.
To simplify the simulations even more,
we may simulate the
covariance matrix of $\hat{\B}$
by averaging the inverse observed information matrices
$\M(\x_{1:n}, \ldots, \x_{k:n})^{-1}$
as indicated in Example~\ref{example:standard-normal-mse-approximate}
and avoid generation of the responses $y_{i}$ 
and calculation of the estimates $\hat{\B}$.
We will use this approach below.

\subsection{Simulation Setup}
\label{subsec:simulation-setup}
\mbox{}

For fixed $k$,
we study the performance of the subsampling scheme 
$\X_{1:n}, \ldots, \X_{k:n}$ of Algorithm~\ref{alg:topk}
based on the $D$-optimal subsampling design
of Theorem~\ref{theorem:dens-of-opt-design-elliptical}
and a simplified version defined in 
Subsection~\ref{subsec:simplified-algorithm} below.
We compare them to other methods
with respect to the $D_{\mathrm{slope}}$-criterion.
The simulations are structured similarly to those 
in~\cite{wang2019information}
to allow for comparison with results in the literature.

In particular,
we consider covariates of dimension $d$ 
equal to fifty.
The covariates are either multivariate normal
or come from a multivariate $t$-distribution
with three degrees of freedom.
The choice of three degrees of freedom is
to maximize the dispersion of the covariates 
while the second moments still exist. 
Both uncorrelated and correlated covariates are considered. 
For the dispersion matrix $\bm{\Sigma}$,
we consider compound symmetry, 
i.\,e.\ $\bm{\Sigma}$ is of the form
$\bm{\Sigma}_{\rho} 
= (1 - \rho) \bm{\mathbb{I}}_{d} + \rho \bm{1}_{d} \bm{1}_{d}^{\top}$ 
with equal correlation $\rho$
between the covariates,
where $\bm{1}_{d} $ denotes a $d$-dimensional vector 
with all entries equal to one. 
In particular, we consider 
the uncorrelated case, $\rho = 0$, 
and a moderate correlation $\rho = 0.5$.

The subdata are of fixed size $k = 1\,000$ 
whereas the size $n$ of the full data varies from 
one thousand to ten millions. 
Note that for $n = 1\,000$
the full data set is selected as subdata
for either method
and that this size is included only for completeness.
The simulations contain $V = 10\,000$ iterations each. 

The simulations are performed as follows:
For each full data size $n$, 
we select subdata based on our approach 
by Algorithm~\ref{alg:topk} (``D-OPT'') 
or its simplified version defined in
Algorithm~\ref{alg:simplified} (``D-OPT-s'') 
and the IBOSS method (``IBOSS'') by \cite{wang2019information}
for comparison. 
Additionally, we select subdata by uniform random subsampling (``UNIF''). 
and compare further to estimates based on the full data (``FULL'') 
to put our approach and the IBOSS method into broader context. 

More precisely,
in each iteration $v$, 
we generate full data of size $n$
and form the $k \times d$ subsample matrix 
$\mathbf{X}^{(v)} = (\X_{1}^{(v)}, \ldots, \X_{k}^{(v)})^{\top}$ 
based on the respective method. 
We calculate the related (conditional) $d \times d$ 
covariance matrix 
$\mathbf{C}_{\mathrm{slope}}^{(v)} 
	= \left({\mathbf{X}^{(v)}}^{\top} \mathbf{X}^{(v)} 
		- k \bar{\X}^{(v)} {\bar{\X}^{(v)}{}}^{\top}\right)^{-1}$
for the slope parameters $\B_{\mathrm{slope}}$, 
where $\bar{\X}^{(v)} = 1/k \sum_{i=1}^{k} \X_{i}^{(v)}$ 
is the mean vector of the subsample. 
We then take the average 
$\mathbf{C}_{\mathrm{slope}} = 1/V \sum_{v = 1}^{V} \mathbf{C}_{\mathrm{slope}}^{(v)}$ 
as the simulated covariance matrix 
for $\hat{\B}_{\mathrm{slope}}$. 
To compare the performance of the methods,
we calculate the determinant of $\mathbf{C}_{\mathrm{slope}}$ 
and standardize it to the homogeneous version 
$\det(\mathbf{C}_{\mathrm{slope}})^{1/d}$. 
This quantity is reported for any of the methods.

\subsection{Simulation Results for Algorithm~\ref{alg:topk}}
\label{subsec:simulation-results-alg-topk}
\mbox{}

Figure~\ref{graphic:NDist} shows the simulation results 
for normally distributed covariates $\X_{i}$ 
with covariance matrices $\bm{\Sigma}_{0} = \Id_{50}$ 
and $\bm{\Sigma}_{0.5} = \frac{1}{2} (\Id_{50} + \bm{1} \bm{1}^{\top})$, 
respectively.
Figure~\ref{graphic:tDist} shows the corresponding results 
for the $t$-distribution with three degrees of freedom 
and the same dispersion matrices
$\bm{\Sigma}_{0}$ and $\bm{\Sigma}_{0.5}$.
In the latter figure, we suppress the uniformly selected subsample 
for focusing on the other methods 
because uniform subsampling performs substantially worse
and the determinant stays close to constant 
at about $4.6 \times 10^{-4}$ for all $n$ in the uncorrelated case 
and at about $8.5 \times 10^{-4}$ in the case with correlation $\rho = 0.5$.
\begin{figure}[htb]
	\centering
	\begin{subfigure}[t]{.475\textwidth}
		\centering
		\includegraphics[width=\linewidth]{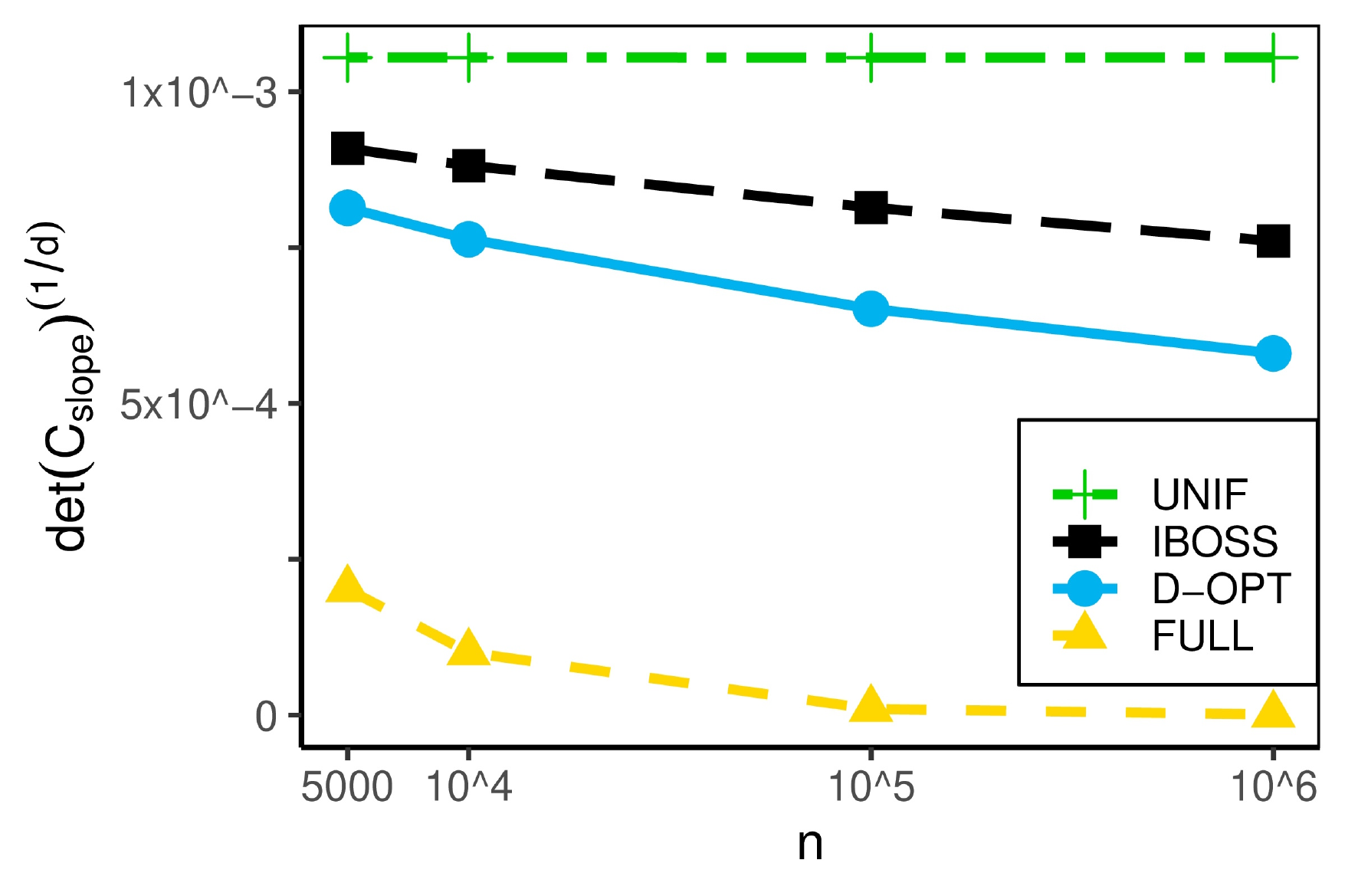}
		\caption{$\X_{i} \sim \Nd\left(\mathbf{0}, \Id_{50}\right)$}
	\end{subfigure}
	\hfill
	\begin{subfigure}[t]{.475\textwidth}
		\centering
		\includegraphics[width=\linewidth]{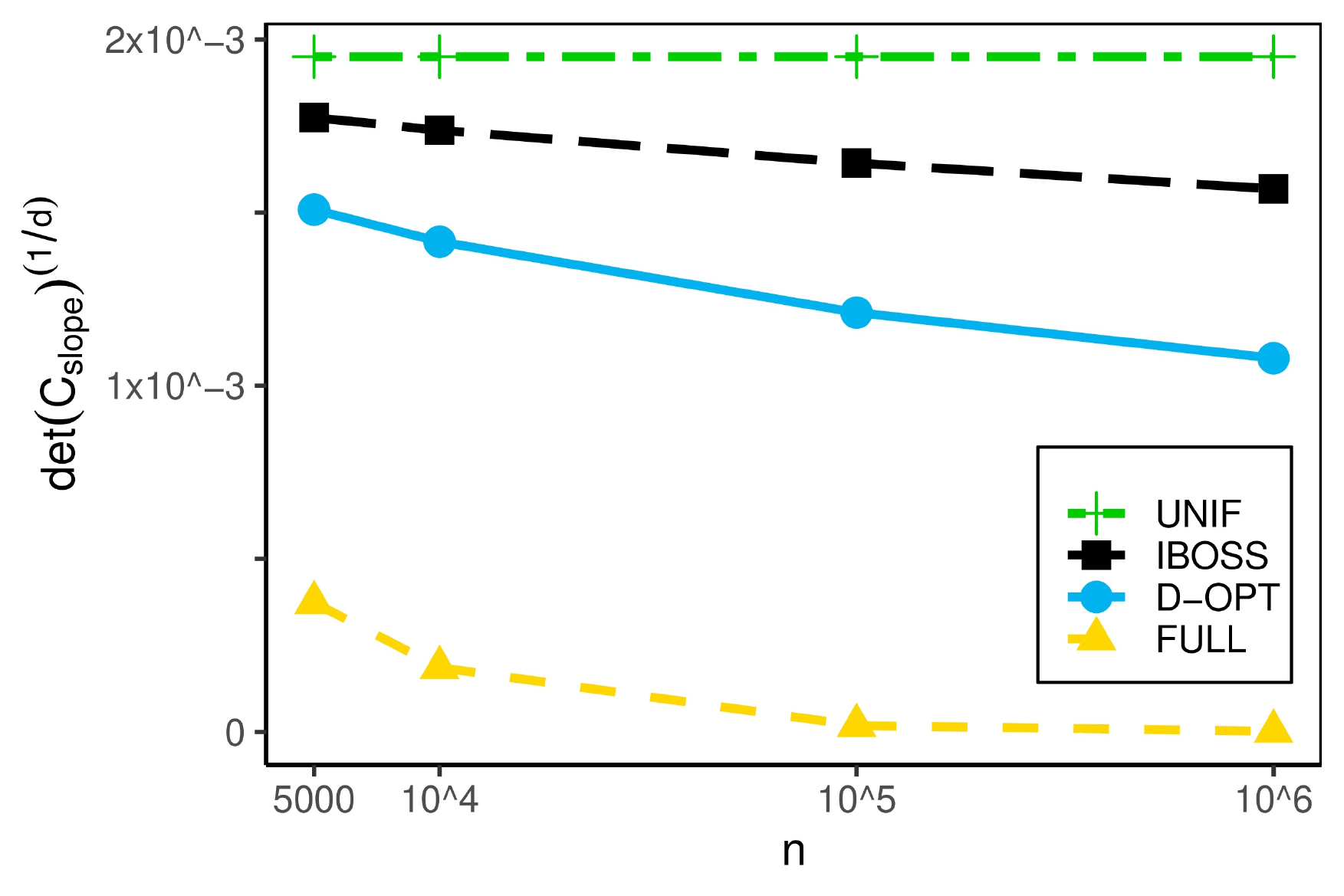}
		\caption{$\X_{i} \sim \Nd\left(\mathbf{0}, \bm{\Sigma}_{0.5}\right)$}
	\end{subfigure}
	\caption{Simulated standardized determinant 
		of the slope covariance matrix 
		for normally distributed covariates, 
		uncorrelated case (left) and correlation $\rho = 0.5$ (right)}
	\label{graphic:NDist}
\end{figure}
\begin{figure}[htb]
	\centering
	\begin{subfigure}[t]{.475\textwidth}
		\centering
		\includegraphics[width=\linewidth]{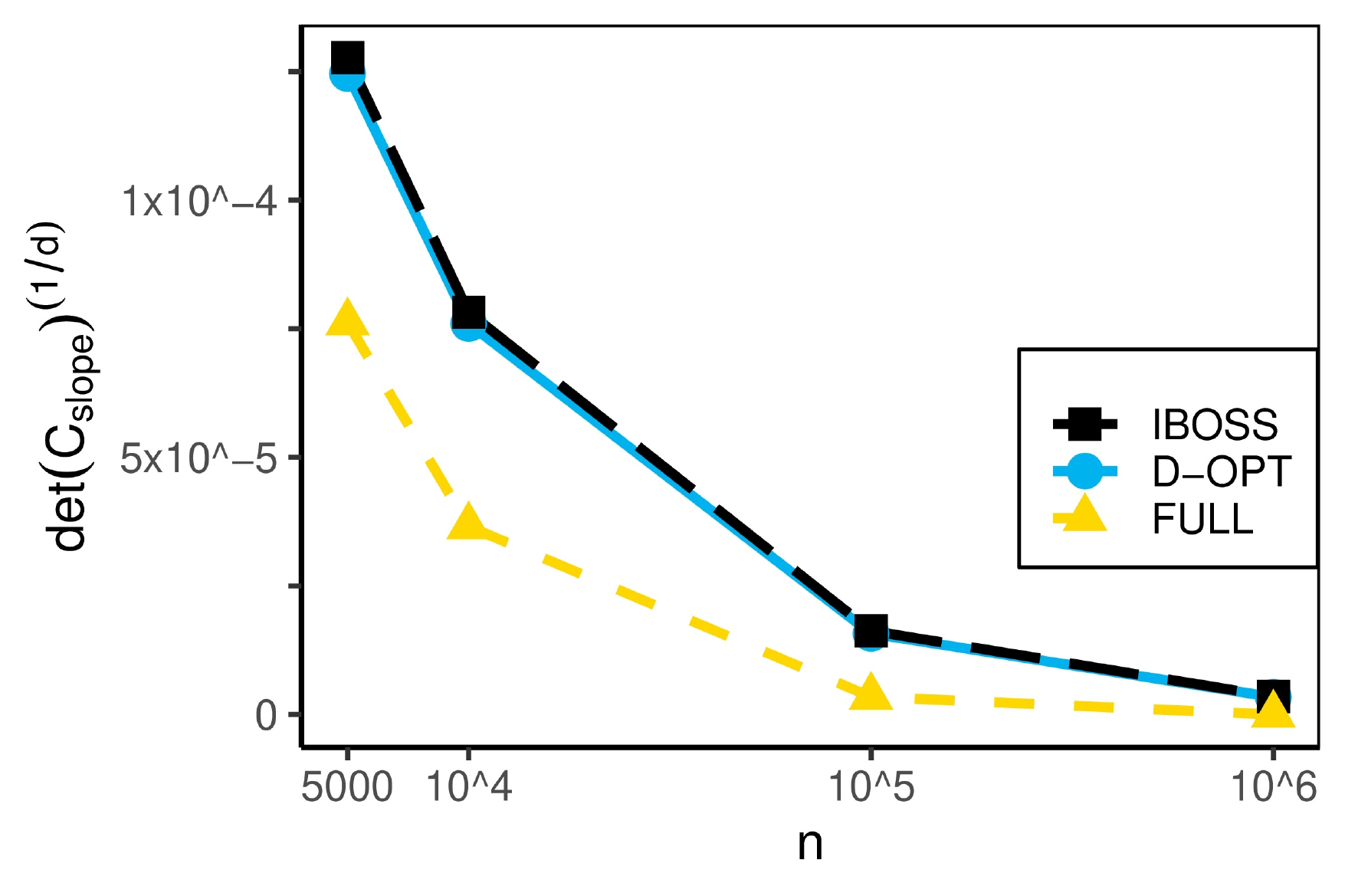}
		\caption{$ \X_{i} \sim t_{3} \left(\mathbf{0}, \Id_{50}\right)$}
	\end{subfigure}
	\hfill
	\begin{subfigure}[t]{.475\textwidth}
		\centering
		\includegraphics[width=\linewidth]{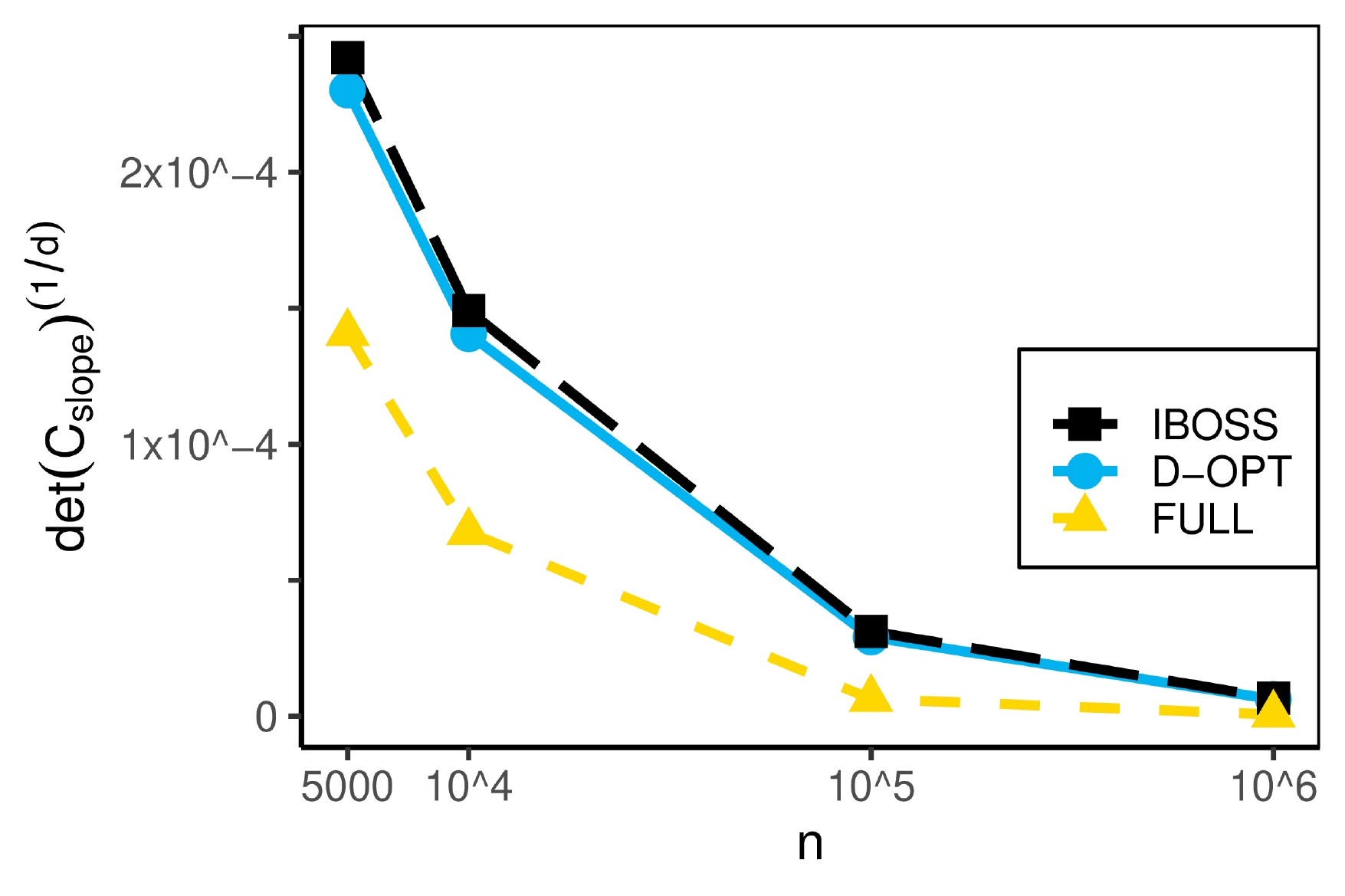}
		\caption{$\X_{i} \sim t_{3} \left(\mathbf{0}, \bm{\Sigma}_{0.5}\right)$}
	\end{subfigure}
	\caption{Simulated standardized determinant 
		of the slope covariance matrix 
		for $t$-distributed covariates with three degrees of freedom, 
		uncorrelated case (left) and correlation $\rho = 0.5$ (right)}
	\label{graphic:tDist}
\end{figure}

As can be seen from the figures, 
our method based on the $D$-optimal subsampling design
is able to outperform the IBOSS method
when the shape of the distribution of the covariates is known. 
Our approach is even more advantageous 
over the IBOSS method when the covariates are correlated. 
In that case, 
the relative efficiency of the IBOSS method with respect to the D-OPT method 
ranges from approximately $0.951$ to $0.928$ 
depending on the full sample size $n$. 
The benefit is however less in the case of the heavy-tailed $t$-distribution 
where both methods perform substantially closer to the full data.
In particular, for large full data size $n$,
both methods work nearly as good as the full data.

\subsection{Computational Complexity}
\label{subsec:mathematical-complexity}
\mbox{}

To judge the computational complexity of 
statistical inference based on subsamples obtained by
the subsampling scheme of Algorithm~\ref{alg:topk}, 
we first notice that the selection 
of $\x_{1:n}, \ldots, \x_{k:n}$ is of order $\Or(nd^{2})$,
where the computation of the inverse 
of the $d \times d$ covariance matrix $\bm{\Sigma}$ is negligible 
for $d \ll n$. 
Computing the least squares estimator $\hat{\B}_{n}$
based on $k$ observations has computational complexity $\Or(kd^{2})$. 
As $k \leq n$, 
the computational complexity is thus $\Or(nd^{2})$
for the entire procedure. 
This is the same order as for computing the least squares estimator on the full data,
but presumably with some smaller constant. 
Because there is no gain in the order of computational complexity,
the subsampling procedure is of practical use only in scenarios, 
where the focus is on the expense of observing the response variable $Y_{i}$,
and not for reducing the computational effort.

\subsection{Simplified Algorithm}
\label{subsec:simplified-algorithm}
\mbox{}

For scenarios where computational complexity is a major issue, 
we, alternatively, propose a simplified method
in which we disregard correlation. 
There we standardize each covariate $X_{ij}$ 
merely by its standard deviation $\sigma_{j}$.

Formally, for transformation of the data, we use the diagonal matrix 
$\tilde{\bm{\Sigma}} = \diag(\sigma_{1}^{2}, \ldots, \sigma_{d}^{2})$ 
containing the diagonal entries of the covariance matrix $\bm{\Sigma}$.
For implementation, we adapt Algorithm~\ref{alg:topk} by replacing 
the Mahalanobis distance 
$\mathrm{d}_{\bm{\Sigma}}(\x_{i}, \mx)$ 
by its simplified counterpart 
$\mathrm{d}_{\tilde{\bm{\Sigma}}}(\x_{i}, \mx) 
	= (\x_{i} - \mx)^{\top} \tilde{\bm{\Sigma}}_{\X}^{-1} (\x_{i} - \mx)$. 
We select those $k$ points
with the largest values of 
$\mathrm{d}_{\tilde{\bm{\Sigma}}}(\x_{i}, \mx)$ 
and denote the resulting subsample by 
$(\tilde{\x}_{1:n}, \ldots, \tilde{\x}_{k:n})$. 
The matrix multiplication in
$\mathrm{d}_{\tilde{\bm{\Sigma}}}(\x_{i}, \mx)$ 
has computational complexity $\Or(nd)$ because 
$\tilde{\bm{\Sigma}}^{-1}$ is a diagonal matrix.  
For a proper subsample,
it is reasonable to assume $k \leq n / d$.
Then the entire subsampling procedure 
has computational complexity $\Or(nd)$, 
\begin{algorithm}
	\caption{Subsample selection according to simplified maximal distance}
	\label{alg:simplified}
	\KwData{Covariates $\x_{i}$, $i = 1, \ldots, n$, mean $\mx$, 
		diagonal matrix $\tilde{\bm{\Sigma}}$ of variances.}
	Fix $k$\;
	Step 1: For $i = 1, \ldots, n$ do: 
	\\
	\phantom{Step 1: }
	Calculate the simplified distance 
	$\mathrm{d}_{\tilde{\bm{\Sigma}}}(\x_{i}, \mx) = (\x_{i} - \mx)^{\top} \tilde{\bm{\Sigma}}^{-1} (\x_{i} - \mx)$\;
	\phantom{Step 1: }
	Repeat\;
	Step 2: Select $\x_{1:n}, \ldots, \x_{k:n}$ 
	corresponding to the $k$ largest values of $\mathrm{d}_{\tilde{\bm{\Sigma}}}(\x_{i}, \mx)$\;
\end{algorithm}

The simplified method has the additional advantage
that it is easier to implement in practice 
when there is no exact knowledge of the covariance matrix of the covariates 
since estimating only the variances on a small uniform random subsample 
(prior to the actual subsampling procedure) 
is much easier than estimating the entire covariance matrix. 
We will see in the simulation study 
that this simplified method is indeed viable. 

We examine the simplified method in the case of normally distributed covariates 
and refer to it as ``D-OPT-s" in the figures. 
First, we note that in the case of uncorrelated covariates, 
the simplified method coincides with D-OPT treated before. 
Thus, in the case of uncorrelated covariates, 
results can be inherited for ``D-OPT-s" 
from Figure~\ref{graphic:NDist}~(A).

In the subsequent simulation, we consider compound symmetry 
of the covariance of the covariates with small ($\rho=0.05$) 
and moderate ($\rho=0.5$) correlation.
Figure~\ref{graphic:oaV} shows the results 
for normally distributed covariates $\X_{i}$ 
with  covariance matrix $\bm{\Sigma}_{0.05}$ and $\bm{\Sigma}_{0.5}$, respectively. 
\begin{figure}[htb]
	\centering
	\begin{subfigure}[t]{.475\textwidth}
		\centering
		\includegraphics[width=\linewidth]{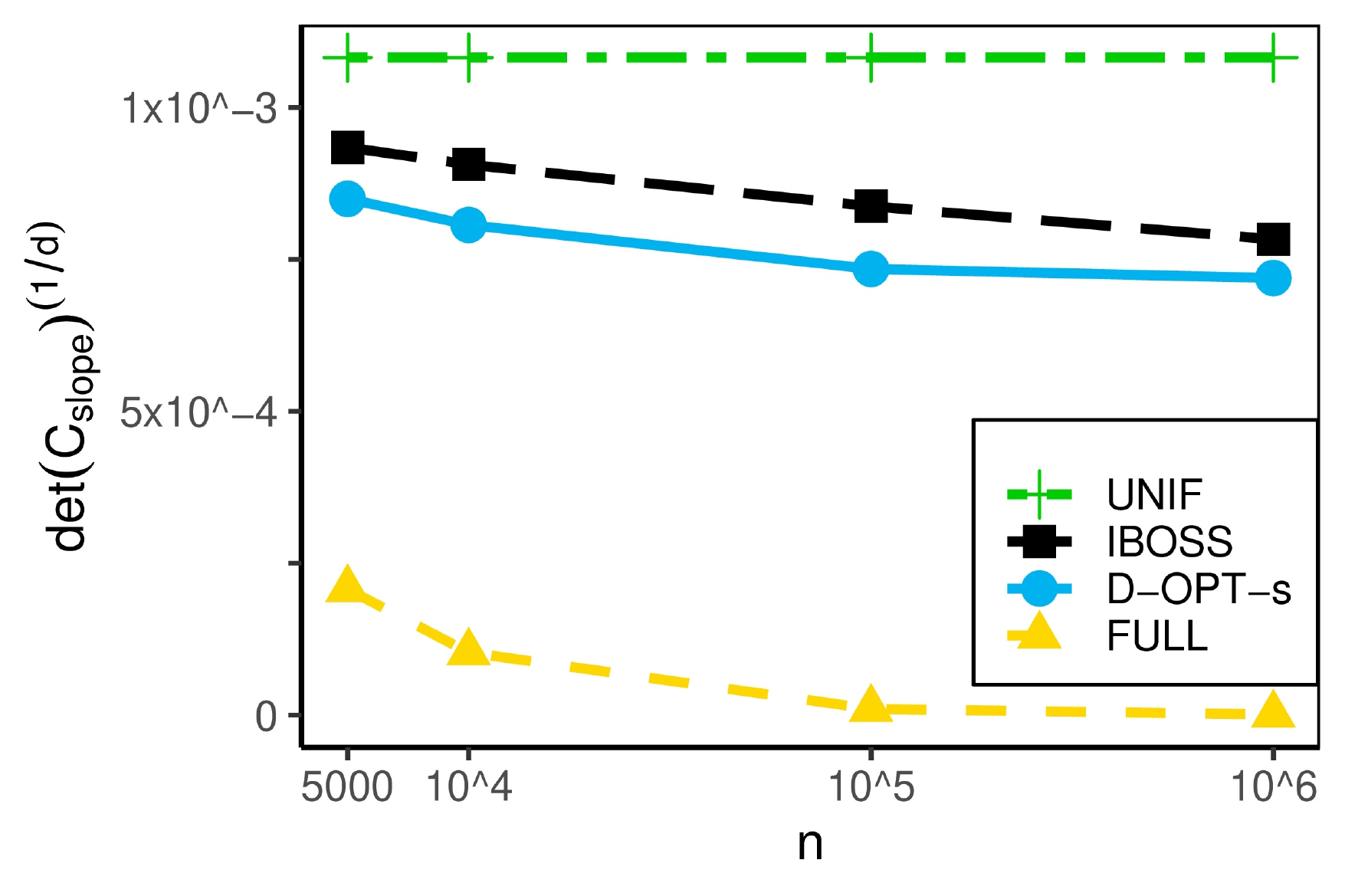}
		\caption{$\X_{i} \sim \Nd\left(\mathbf{0}, \bm{\Sigma}_{0.05} \right)$}
	\end{subfigure}
	\hfill
	\begin{subfigure}[t]{.475\textwidth}
		\centering
		\includegraphics[width=\linewidth]{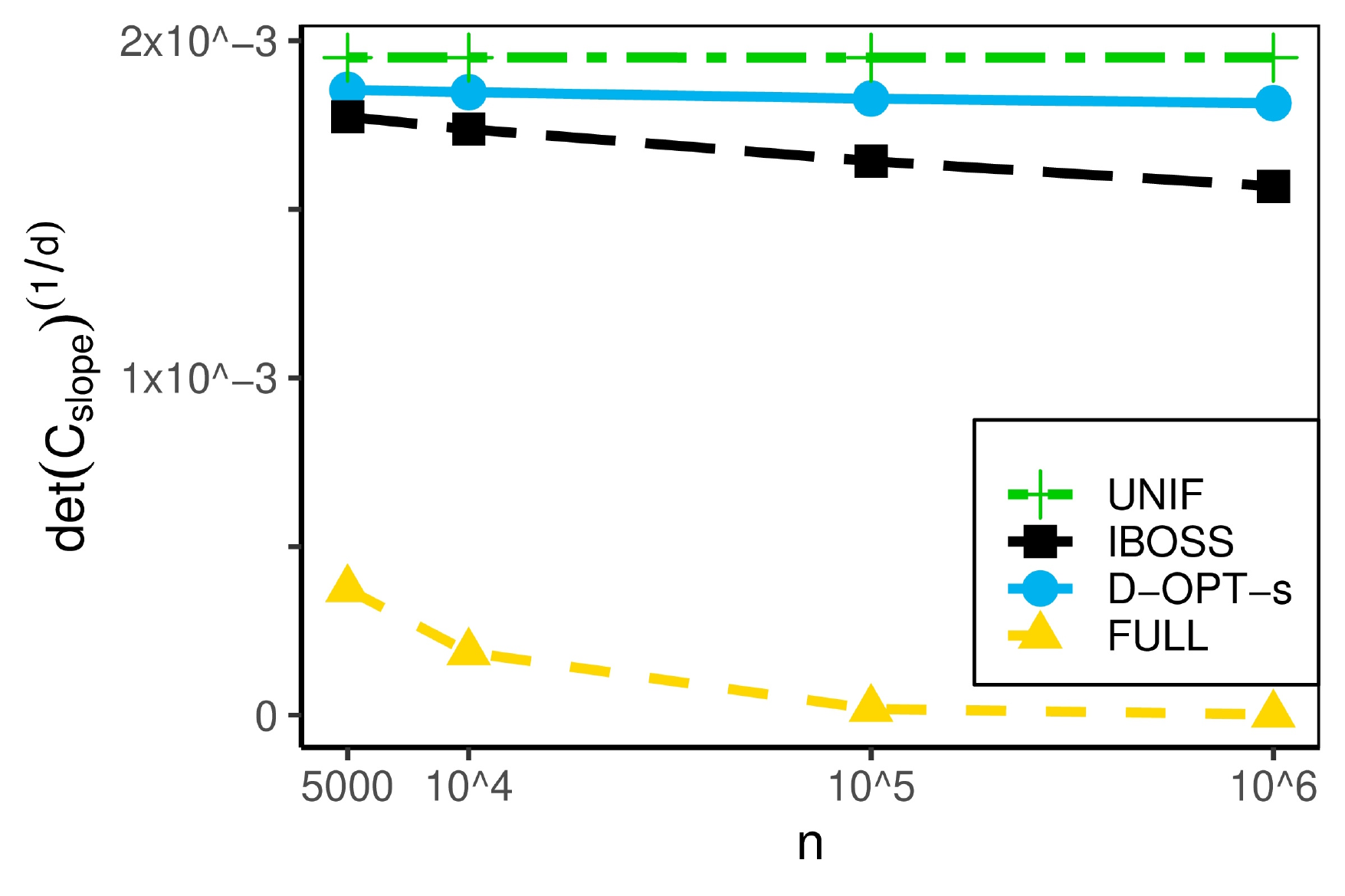}
		\caption{$\X_{i} \sim \Nd\left(\mathbf{0}, \bm{\Sigma}_{0.5} \right)$}
	\end{subfigure}
	\caption{Simulated standardized determinant 
		of the slope covariance matrix 
		for the simplified D-OPT-s method
		in the case of normally distributed covariates,
		small (left) and moderate correlation (right)}
	\label{graphic:oaV}
\end{figure}
While  the advantage of the $D$-optimal subsampling design  
over the IBOSS method seems to be reduced, 
there are still scenarios where D-OPT-s can outperform the IBOSS method 
as, for example, in the case of the covariance matrix $\bm{\Sigma}_{0.05}$ 
when the correlation is small.
However, if correlation is larger, 
as in the case of the covariance matrix $\bm{\Sigma}_{0.5}$, 
the simplified method D-OPT-s seems to perform inferior to IBOSS
and only slightly better than uniform random subsampling. 

For quantification of the variability in the simulation, 
we also report the standard deviation 
alongside with the mean 
of the standardized determinant 
$\det\left(\mathbf{C}_{\mathrm{slope}}^{(v)}\right)^{1/d}$ 
of the simulated slope covariance matrix $\mathbf{C}_{\mathrm{slope}}^{(v)}$
for all five methods in Table~\ref{table:all-methods}. 
\begin{table}[h]
	\begin{center}
		\caption{Mean and standard deviation 
			of the standardized determinant 
			$\det\left(\mathbf{C}_{\mathrm{slope}}^{(v)}\right)^{1/d}$ 
			of the simulated slope covariance matrix 
			for covariates $\X_{i} \sim \Nd\left(\mathbf{0}, \bm{\Sigma}_{0.5} \right)$ 
			and full sample size $n = 10^4$, $n = 10^6$}
		\begin{tabular}{l|l|ccccc} \toprule
			{$n$} & {} & {FULL} & {D-OPT} & {D-OPT-s} & {IBOSS} & {UNIF} 
			\\ 
			\midrule
			\multirow{2}{*}{$10^4$} & mean 
				& $1.854 \times 10^{-4}$ & $1.380 \times 10^{-3}$ & $1.799 \times 10^{-3}$ 
				& $1.693 \times 10^{-3}$ & $1.899 \times 10^{-3}$ 
			\\ 
			& std 
				& $3.736 \times 10^{-7}$ & $5.161 \times 10^{-6}$ & $1.142 \times 10^{-5}$ 
				& $1.030 \times 10^{-5}$ & $1.226 \times 10^{-5}$ 
			\\ 
			\midrule
			\multirow{2}{*}{$10^6$} & mean 
				& $1.849 \times 10^{-6}$ & $1.052 \times 10^{-3}$ & $1.768 \times 10^{-3}$ 
				& $1.529 \times 10^{-3}$ & $1.899 \times 10^{-3}$ 
			\\
			& std 
				& $3.689 \times 10^{-10}$ & $2.471 \times 10^{-6}$ & $1.118 \times 10^{-5}$ 
				& $8.814 \times 10^{-6}$ & $1.225 \times 10^{-5}$ 
			\\ 
			\bottomrule
		\end{tabular}
		\label{table:all-methods}
	\end{center}
\end{table}
Here, we consider again normally distributed covariates of dimension $d = 50$ 
and moderate correlation ($\bm{\Sigma}_{0.5}$). 
Finally, for the same setting, we showcase the computing times
of the simulations in milliseconds 
for the D-OPT, D-OPT-s, and IBOSS methods, respectively,
 in Figure~\ref{graphic:boxplot_ct}. 
\begin{figure}[htb]
	\centering
	\includegraphics[width=\linewidth]{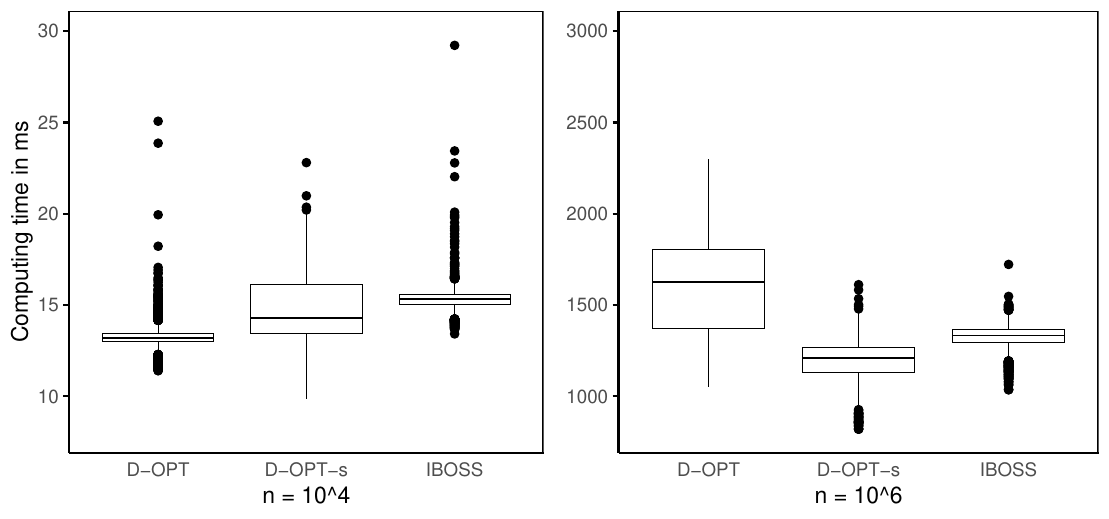}
	\caption{Computing times of the simulations 
		for full data size $n = 10^4$ (left) and $n = 10^6$ (right), 
		$\X_{i} \sim \Nd\left(\mathbf{0}, \bm{\Sigma}_{0.5} \right)$}
	\label{graphic:boxplot_ct}
\end{figure}
We find that the D-OPT-s method is consistently faster than the IBOSS method 
in our simulations, even though both methods 
share the same computational complexity $\Or(nd)$.

\section{Discussion}
\label{sec:discussion}
In the present paper, we have characterized 
$D$-optimal subsampling designs $\xi_{\alpha}^{*}$
 for multiple linear regression, 
 first for centered spherical distributions. 
Then, we have extended the characterization to elliptical distributions 
by location-scale transformations. 
Thereby, we have generalized the results in \cite{reuter2023optimal} 
on ordinary linear regression to multiple covariates.

We have presented two different methods of  
subsampling and discussed their computational complexity. 
The D-OPT method based on the Mahalanobis distance 
with respect to the full covariance matrix
has complexity of order $\Or(nd^{2})$
whereas the simplified version D-OPT-s neglecting correlation 
can be performed with a computational complexity of order $\Or(nd)$. 
We have compared both methods with IBOSS 
proposed by~\cite{wang2019information} in simulation studies. 
These simulations illustrate the expected property 
that the full D-OPT method outperforms IBOSS.
Further, the simplified method D-OPT-s may perform better than IBOSS in settings 
when the correlation between covariates is small,
but may be less efficient when the correlation becomes larger. 

In addition to the simulation of
the standardized determinant
$\det(\mathbf{C}_{\mathrm{slope}})^{1/d}$
based on the \textit{observed} information matrices,
we have also simulated the mean squared error 
of the slope estimates $\hat{\B}_{\mathrm{slope}}$ by
$1/V \sum_{v = 1}^{V} \Vert \hat{\B}_{\mathrm{slope}}^{(v)} - \B_{\mathrm{slope}} \Vert^{2}$, 
to compare the different methods
with each other.
In all cases, 
results were very similar 
to those for $\det(\mathbf{C}_{\mathrm{slope}})^{1/d}$
and, what is more important, 
the ranking in the performance
of the different methods does not change. 
Beside applications where the covariance matrix of the covariates is known, 
the full method can be used as a benchmark for other methods 
proposed in the literature.

To construct subsamples in real data situations
according to the full D-OPT method,
those units are selected which have largest 
Mahalanobis distance $\mathrm{d}_{\bm{\Sigma}}(\x, \mx)$
from the mean.
Thus, only the mean $\mx$
and the dispersion matrix $\bm{\Sigma}$
of the underlying elliptical distribution have to be known
to create the subsample.
When the mean $\mx$
and the dispersion $\bm{\Sigma}$
are not known in advance, they may be substituted by 
their empirical counterparts
$\bar{\x} = \frac{1}{n} \sum_{i=1}^{n} \x_{i}$
and
$\mathbf{S}_{\x} = \frac{1}{n} \sum_{i=1}^{n} (\x_{i} - \bar{\x}) (\x_{i} - \bar{\x})^{\top}$
of the full data.
The resulting observed Mahalanobis distance 
$\mathrm{d}_{\mathbf{S}_{\x}}(\x_{i}, \bar{\x})
	= (\x_{i} - \bar{\x})^{\top} \mathbf{S}_{\x}^{-1} (\x_{i} - \bar{\x})$
differs from the leverage 
$h_{i} = \f(\x_{i})^{\top} \M(\x_{1}, \ldots, \x_{n})^{-1} \f(\x_{i})$
only by some constants,
$h_{i} 	= (\mathrm{d}_{\mathbf{S}_{\x}}(\x_{i}, \bar{\x}) + 1) / n$,
where
$\M(\x_{1}, \ldots, \x_{n}) = \sum_{i=1}^{n} \f(\x_{i}) \f(\x_{i})^{\top}$
is the observed information matrix of the full data.
Hence, selecting the units with largest
observed Mahalanobis distance
$\mathrm{d}_{\mathbf{S}_{\x}}(\x_{i}, \bar{\x})$
is equivalent to selecting those units
with highest leverage $h_{i}$.

Note that this approach differs from the method
of subsampling via algorithmic leveraging 
as described in~\cite{ma2014statistical} 
where a sampling distribution 
proportional to the leverage scores $h_{i}$ is used
with replacement.
Hence, algorithmic leveraging does not fit 
into the present framework of subsampling designs
and may suffer from the undesirable property 
of multiple selection of the same units.

If the empirical mean $\bar{\x}$ and dispersion $\mathbf{S}_{\x}$
are not readily available, they may be replaced 
in Algorithm~\ref{alg:topk}
by estimates based on a prior random subsample of the full data.

For other convex, differentiable optimality criteria 
like Kiefer's $\Phi_{q}$-criteria
of matrix means
including the $A$-criterion for $q = 1$
\citep [see, e.\,g.,][Chapter~6]{pukelsheim1993optimal}, 
corresponding versions of
Theorem~\ref{theorem:opt-design-general-distribution}
apply.
In particular,
when the covariates have a
centered spherical distribution, 
these criteria are also rotationally invariant.
Hence, the $D$-optimal subsampling design of
Theorem~\ref{theorem:dens-of-opt-design-spherical}
is also optimal with respect to any $\Phi_{q}$-criterion.
Thus, real subsamples generated
by Algorithm~\ref{alg:topk}
meet these criteria, too.

However, these criteria are not equivariant
with respect to linear transformations, in general,
so that the $D$-optimal subsampling design
will no longer be optimal for other criteria.
Although, by the corresponding equivalence theorems,
the particular optimal subsampling design 
will accept all units outside some ellipsoid
as in Theorem~\ref{theorem:opt-design-general-distribution},
the scaling matrix defining the ellipsoid will differ.
For example, in the case of the $\MSE$ of the slope estimates 
($A_{\mathrm{slope}}$-criterion)
considered in Section~\ref{sec:method}
when the distribution is elliptical,
the $A_{\mathrm{slope}}$-optimal subsampling design $\xi_{\alpha}^{*}$
has density
$f_{\xi_{\alpha}^{*}}(\x) = f_{\X}(\x)$ 
for $(\x - \mx)^{\top} \mathbf{S}(\xi_{\alpha}^{*})^{-2} (\x - \mx) \geq q_{1 - \alpha}(\xi_{\alpha}^{*})$,
and $f_{\xi_{\alpha}^{*}}(\x) = 0$ otherwise,
where $q_{1 - \alpha}(\xi_{\alpha}^{*})$ is the $(1 - \alpha)$-quantile of the distribution of 
$(\X_{i} - \mx)^{\top} \mathbf{S}(\xi_{\alpha}^{*})^{-2} (\X_{i} - \mx)$.
In contrast to the situation
of Theorem~\ref{theorem:dens-of-opt-design-elliptical}
for $D$-optimality,
the boundary
$\{\x; \, (\x - \mx)^{\top} \mathbf{S}(\xi_{\alpha}^{*})^{-2} (\x - \mx) = q_{1 - \alpha}\}$
is not a contour of the density $f_{\X}$
when the distribution is not spherical.
Then both the scaling matrix $\mathbf{S}(\xi_{\alpha}^{*})^{2}$ 
and the quantile $q_{1 - \alpha}(\xi_{\alpha}^{*})$
will be difficult to be determined.

As an alternative, we may consider 
the expected mean squared error (EMSE) criterion
$\mathrm{EMSE}(\xi) = \int \f(\x)^{\top} \M(\xi)^{-1} \f(\x) \, f_{\X}(\x) \diff \x$
which measures the average of the prediction variance
$\Var[\f(\x)^{\top} \hat{\B}] = \f(\x)^{\top} \M(\xi)^{-1} \f(\x)$
for estimating the mean response $\E[Y(\x)] = \f(\x)^{\top} \B$ 
of further observations $Y(\x)$ at $\x$,
where the average is taken
according to the distribution of the covariates $\X_{i}$.
Similar to the $D$-criterion, 
the EMSE-criterion
is equivariant with respect to linear transformations.
In particular,
when the covariates have a
centered spherical distribution, 
the $D$-optimal subsampling design of
Theorem~\ref{theorem:dens-of-opt-design-spherical}
is seen to be EMSE-optimal.
Then, by equivariance,
the $D$-optimal subsampling design $\xi_{\alpha}^{*}$
of Theorem~\ref{theorem:dens-of-opt-design-elliptical}
is also EMSE-optimal for elliptical distributions,
and Algorithm~\ref{alg:topk} 
provides a suitable method
to  generate real subsamples with minimal
expected prediction variance.
These findings may be
readily extended to the general class 
of criteria based on powers of the prediction variance
by \cite{dette1999predictedvariance}
when averaging is
according to the distribution of the covariates.
For a recent study on subsampling with a focus on prediction error, see the work by \cite{ciamina2025}. The authors introduce a new optimality criterion that extends the goal of minimizing the Random–X prediction error by also accounting for the joint distribution of the covariates.

For multiple quadratic regression,
$Y_{i} = \beta_{0} + \sum_{j=1}^{d} \beta_{j} X_{ij} 
					+ \sum_{j=1}^{d} \beta_{jj} X_{ij}^{2}
					+ \sum_{j<j^{\prime}} \beta_{jj^{\prime}} X_{ij} X_{ij^{\prime}} 
					+ \e_{i}$,
invariance and equivariance considerations
may be used as in the case of multiple linear regression
\citep[see, e.\,g.,][Chapter~15]{pukelsheim1993optimal}.
Similar to the results in one dimension 
by \cite{reuter2023optimal},
$D$-optimal subsampling designs $\xi_{\alpha}^{*}$ 
may be obtained which have density 
$f_{\xi_{\alpha}^{*}}(\x) = f_{\X}(\x)$ 
for $\mathrm{d}_{\bm{\Sigma}}(\x, \mx) \leq q_{\alpha_{1}}$
or $\mathrm{d}_{\bm{\Sigma}}(\x, \mx) \geq q_{1 - \alpha_{2}}$,
and $f_{\xi_{\alpha}^{*}}(\x) = 0$ otherwise,
where $q_{\alpha_{1}}$ and $q_{1 - \alpha_{2}}$
are suitable quantiles of the Mahalanobis distance
$\mathrm{d}_{\bm{\Sigma}}(\X_{i}, \mx)$
satisfying $\alpha_{1} + \alpha_{2} = \alpha$
and a second, nonlinear equation 
arising from the equivalence theorem
(Theorem~\ref{theorem:opt-design}).
Hence, in real subsampling,
those units will be selected
which either have a large or which have a small Mahalanobis distance
$\mathrm{d}_{\bm{\Sigma}}(\x_{i}, \mx)$
to the mean.
As in the multiple linear case,
the quantiles $q_{\alpha_{1}}$ and $q_{1 - \alpha_{2}}$
do not depend on the location and scaling 
parameters $\mx$ and $\bm{\Sigma}$.
in particular,
for multivariate normal distribution of the covariates,
$q_{\alpha^{\prime}} = \chi_{d, \alpha^{\prime}}^{2}$ 
is again the $\alpha^{\prime}$-quantile of the $\chi^{2}$-distribution 
with $d$ degrees of freedom.
In contrast to that,
the partial proportions $\alpha_{1}$ and $\alpha_{2}$
vary with the distribution of the covariates.
as exhibited in \cite{reuter2023optimal} 
in the case of a single covariate ($d = 1$).
In particular,
the interior region may vanish ($\alpha_{1} = 0$)
for heavy-tailed distributions
while the exterior region is always required ($\alpha_{2} > 0$).
Also, for higher order polynomials,
the structural results 
by \cite{reuter2023optimal}
can be extended to multiple covariates:
When the polynomial model contains 
all terms up to order $q$,
the $D$-optimal subsampling design
is concentrated on, at most, $(q + 1) / 2$ 
concentric elliptical shells 
when $q$ is odd,
and on, at most, $(q + 2) / 2$ 
concentric elliptical shells 
when $q$ is even.

The optimal subsampling designs 
considered in the present paper
depend both on the distribution of the covariates 
and on the model 
which relates the response variable $y_{i}$ to the covariates $\X_{i}$. 
If either of them is not correctly specified, 
the proposed subsampling designs will no longer be optimal.
Related work 
on subsampling for model discrimination
is done by~\cite{yu2022subdata}.

\section*{Acknowledgments}
The work of the first author is supported by the Deutsche Forschungsgemeinschaft (DFG, German Research Foundation) - 314838170, GR 2297 Math Core.
The authors are grateful to Norbert Gaffe
for communicating the proof of Lemma~\ref{lemma:asymptotic-normality}.

\appendix 

\section{Technical Details}

Denote by $\1_{A}$ the indicator function on a set $A$.

For asymptotic properties, 
we consider sequences of random variables.

\begin{lemma}
	\label{lemma:asymptotic-normality}
	Let $Y_{i} = \f(\X_{i})^{\top} \B + \e_{i}$
	be a general linear model in $p$ parameters 
	with i.\,i.\,d.\ covariates $\X_{i}$ 
	satisfying $\E[\Vert \f(\X_{i}) \Vert^{2}] < \infty$
	and i.\,i.\,d.\ observational errors $\e_{i}$
	with variance $\sigma_{\e}^{2}$,
	$i \geq 1$, independent of each other.
	Let $\hat{\B}_{n}$ be the least squares estimator based on a subsample
	of $(Y_{1}, \X_{1}), \ldots, (Y_{n}, \X_{n})$ 
	generated according to a continuous subsampling design $\xi$
	with positive definite information matrix
	$\M(\xi) = \int \f(\x) \f(\x)^{\top} f_{\xi}(\x) \diff \x$.
	Then
	\[
		\sqrt{n}(\hat{\B}_{n} - \B) \stackrel{\mathcal{D}}{\to} 
			\Nd_{p}\left(\mathbf{0}, \sigma_{\e}^{2} \M(\xi)^{-1}\right) \, .
	\]
\end{lemma}

\begin{proof}[Proof (following \cite{gaffke2024asymptotics})]
	For $f_{\X}(\x) > 0$,
	let $\varphi(\x) = f_{\xi}(\x) / f_{\X}(\x)$
	be the conditional probability for selecting
	a unit $i$ when $\X_{i} = \x$,
	and let $\varphi(\x) = 0$ otherwise.
	To practically generate a subsample,
	let $U_{i}$, $i \geq 1$, be a sequence of i.\,i\,d.\ 
	random variables uniform on $[0, 1]$,
	independent of all $\X_{i}$ and $\e_{i}$.
	Set $Z_{i} = \1_{U_{i} \leq \varphi(\X_{i})}$.
	Then $Z_{i}$ is a Bernoulli variable with
	success probability $\alpha$,
	and the subsample can be generated 
	by selecting those units $i$
	for which $Z_{i} = 1$.
	
	The least squares estimator $\hat{\B}_{n}$
	based on the subsample
	can be defined to minimize
	$\sum_{i=1}^{n} Z_{i} (Y_{i} - \f(\X_{i})^{\top} \B)^{2}$.
	For $n$ large enough,
	\begin{align*}
		\hat{\B}_{n}
			& = \left(\sum_{i=1}^{n} Z_{i} \f(\X_{i}) \f(\X_{i})^{\top}\right)^{-1}
						\sum_{i=1}^{n} Z_{i} \f(\X_{i}) Y_{i}
			& = \B + \left(\sum_{i=1}^{n} Z_{i} \f(\X_{i}) \f(\X_{i})^{\top}\right)^{-1}
						\sum_{i=1}^{n} Z_{i} \f(\X_{i}) \e_{i} \, .
	\end{align*}
	By the Strong Law of Large Numbers, we obtain
	\[
		\frac{1}{n} \sum_{i=1}^{n} Z_{i} \f(\X_{i}) \f(\X_{i})^{\top}
			\to \E[\varphi(\X_{i}) \f(\X_{i}) \f(\X_{i})^{\top}]
					= \M(\xi)
	\]
	almost surely.
	Further, 
	$\E[Z_{i} \f(\X_{i}) \e_{i}] = \mathbf{0}$ and
	$\Cov[Z_{i} \f(\X_{i}) \e_{i}]  
		= \E[Z_{i} \f(\X_{i}) \f(\X_{i})^{\top} \e_{i}^{2}] 
		= \sigma_{\e}^{2} \M(\xi)$.
	Hence,
	by the multivariate Central Limit Theorem,
	we get
	\[
		\frac{1}{\sqrt{n}} \sum_{i=1}^{n} Z_{i} \f(\X_{i}) \e_{i}
			\stackrel{\mathcal{D}}{\to} 
				\Nd_{p}\left(\mathbf{0}, \sigma_{\e}^{2} \M(\xi)\right) \, .
	\]
	Then, the result follows by Slutsky's theorem.
\end{proof}

For stating the equivalence theorem
to characterize $D$-optimality,
we introduce the sensitivity function
\begin{equation}
	\label{eq:sensitivity}
	\psi(\x, \xi) = \alpha \f(\x)^{\top} \M(\xi)^{-1} \f(\x) 
\end{equation}
of a subsampling design $\xi$.
The sensitivity function $\psi(\x, \xi)$ constitutes 
the essential part of the directional derivative
of the $D$-criterion in the direction of a one-point design $\xi_{\x}$ 
with total mass $\alpha$ at $\x$.
Note that $\xi_{\x}$  
is not a continuous subsampling design itself.
Similar to Theorem~3.1. in \cite{reuter2023optimal},
we can paraphrase 
Corollary~1~(c) in \cite{sahm2001note} 
for the present purposes.

\begin{theorem}
	\label{theorem:opt-design}
	Let $\xi_{\alpha}^{*}$ be a subsampling design 
	and let the distribution of $\psi(\X_{i}, \xi_{\alpha}^{*})$ be continuous.
	Then $\xi_{\alpha}^{*}$ is $D$-optimal if and only if
	there exists $c^{*}$ such that
		\[	
			f_{\xi_{\alpha}^{*}}(\x) = f_{\X}(\x) \1_{\{\psi(\x, \xi_{\alpha}^{*}) \geq c^{*}\}} \, .
		\]	
\end{theorem}

\begin{proof}[Proof of Theorem~\ref{theorem:opt-design-general-distribution}]
	In the multiple linear regression
	model~\eqref{eq:model-multiple-regression},
	the sensitivity function~\eqref{eq:sensitivity}
	can be rewritten as
	\begin{equation*}
		\label{eq:sensitivity-linear}
		\psi(\x, \xi) = \alpha (\x - \m(\xi))^{\top} \mathbf{S}(\xi)^{-1} (\x - \m(\xi)) + 1 
	\end{equation*}
	and is a quadratic form in $\x$ 
	(up to the additive constant $1$).
	For each $s$, the level set $\{\psi(\x, \xi) = s\}$
	is, at most, the surface of an ellipsoid
	and has Lebesgue measure zero.
	Thus the continuity condition on the distribution
	of $\psi(\X_{i}, \xi)$ is satisfied,
	and the result follows from Theorem~\ref{theorem:opt-design}.
\end{proof}

\begin{proof}[Proof of Corollary~\ref{corollary:opt-design-one-cov-general}]
 	For $d = 1$,
 	the sensitivity function 
 	$\psi(x, \xi) = \alpha (x - m(\xi))^{2} / s^{2}(\xi) + 1$
 	is a polynomial of degree two in $x$ 
 	with positive leading term,
 	where $s^{2}(\xi) = \int x^{2} f_{\xi}(x) \diff x - \alpha m(\xi)^{2}$.
 	The support $\{\psi(x, \xi_{\alpha}^{*}) \geq c^{*}\}$
 	of the $D$-optimal subsampling design $\xi_{\alpha}^{*}$ 
 	reduces to the exterior of an interval $(a, b)$
 	which is symmetric with respect to 
 	$\m(\xi_{\alpha}^{*}) = \alpha^{-1} \int x f_{\xi_{\alpha}^{*}}(x) \diff x$.
 	Further, $\Prob(X_{i} \leq a \mbox{ or } X_{i} \geq b) = \alpha$
 	because $\xi_{\alpha}^{*}$ is a 
 	subsampling design of proportion $\alpha$.
\end{proof}

To extend the concept of symmetrization 
to multiple covariates ($d \geq 2$),
we notice that
the regression model is linearly equivariant 
with respect to 
affine linear transformations 
$\g_{\bm{A}, \bm{\mu}}(\x) = \bm{A} \x + \bm{\mu}$
of the covariates as
\begin{equation*}
	\label{eq:multiple-regression-equivariance-affine-linear}
	\f(\g_{\bm{A}, \bm{\mu}}(\x)) = \Q_{\bm{A}, \bm{\mu}} \f(\x) \, ,
	\qquad
	\Q_{\bm{A}, \bm{\mu}} 
		= \begin{pmatrix}
				1 & \mathbf{0}         
				\\
				\bm{\mu} & \bm{A} 
			\end{pmatrix} \, ,
\end{equation*}
with nonsingular transformation matrix $\bm{A}$.
In particular,
the model is linearly equivariant 
with respect to rotations $\g \in SO(d)$ as
\begin{equation}
	\label{eq:multiple-regression-equivariance-sod}
	\f(\g(\x)) = \Q_{\g} \f(\x) \, ,
	\qquad
	\Q_{\g} 
		= \begin{pmatrix}
				1 & \mathbf{0}         
				\\
				\mathbf{0} & \Rho_{\g} 
			\end{pmatrix}
			\f(\x) \, ,
\end{equation}
where $\Rho_{\g}$ is the orthogonal rotation matrix on $\x$
corresponding to $\g$, 
i.\,e.\ $\g(\x) = \Rho_{\g} \x$,
so that the transformation matrix $\Q_{\g}$ of 
the regression function $\f$ has determinant one.
Further,
$\g$ induces the transformation  
$\M(\xi^{\g}) = \Q_{\g} \M(\xi) \Q_{\g}^{\top}$
of the information matrix,
where $\xi^{\g}$ denotes the image
of $\xi$ under $\g$.
Hence, the $D$-criterion is invariant 
with respect to transformations $\g \in SO(d)$, 
$\det(\M(\xi^{\g})) = \det(\M(\xi))$.

To make use of the rotational invariance,
we consider the representation of $\R^{d}$ 
in hyperspherical coordinates
$(r, \tht) \in [0,\infty) \times \mathbb{B}$,
where 
$\mathbb{B} = [0, \pi)^{d-2} \times [0, 2 \pi)$
is the sample space
of the angular vector
$\tht = (\theta_{1}, \ldots, \theta_{d-1})^{\top}$.
For matching Cartesian and hyperspherical representation,
we can use the transformation 
$\bm{T}: [0,\infty) \times \mathbb{B} \to \R^{d}$, 
$\bm{T}(r, \tht) = \x$, 
where 
$x_{k} = r \cos(\theta_{k}) \prod_{j=1}^{k-1} \sin(\theta_{j})$, 
$k = 1, \ldots, d-1$, 
and 
$x_{d} = r \prod_{j=1}^{d-1} \sin(\theta_{j})$. 
We identify all points of radius zero 
with the origin ($\x = \mathbf{0}$)
and denote the inverse of the transformation $\bm{T}$ 
by $\bm{U} = \bm{T}^{-1}$. 
	
For any subsampling design $\xi$ on $\R^{d}$, 
the induced subsampling design 
$\xi^{\bm{U}} = \xi_{(R, \Th)}$
is a joint design on the radius $r$
and the angles $\tht$
in hyperspherical coordinates. 

By the Radon-Nikodym theorem,
the design $\xi_{(R, \Th)}$
can be decomposed into the 
measure theoretic product 
$\xi_{R} \otimes \xi_{\Th | R}$ 
of the marginal subsampling design $\xi_{R}$ 
of mass $\alpha$ on the radius  
and the conditional design 
$\xi_{\Th | R=r}$ 
on the vector $\theta$ of angles 
given the radius $R = r$.
By standardization of the conditional design
$\xi_{\Th | R=r}$ 
as a Markov kernel,
the sample space
$\mathbb{B}$
of the angles has mass one,
$\xi_{\Th | R=r}(\mathbb{B}) = 1$, 
for any radius $r$. 
It follows from $f_{\xi} \leq f_{\X}$
that the density $f_{R}$ of the marginal design $\xi_{R}$ is bounded 
by the marginal density $f_{R(\X)}$ of $\X_{i}$ on the radius.

\begin{lemma}
	\label{lemma:decomposition}
	$\xi$ is invariant with respect to $SO(d)$ 
	if and only if $\xi^{\bm{U}} = \xi_{R} \otimes \bar{\mu}$.
\end{lemma}

\begin{proof}
	This follows from the fact
	that $\bar{\mu}$ is the unique invariant measure
	of mass one on $\mathbb{B}$
	and that the Borel $\sigma$-algebra 
	on $[0,\infty) \times \mathbb{B}$
	is the product $\sigma$-algebra of
	the Borel $\sigma$-algebras on $[0,\infty)$ and $\mathbb{B}$,
	respectively.
\end{proof}

\begin{lemma}
	\label{lemma:sym-design-bounded}
	Let the covariates $\X_{i}$ have a centered spherical distribution.
	If the design $\xi$ has density $f_{\xi} \leq f_{\X}$,
	then its symmetrization 
	$\bar{\xi} = \xi_{R} \otimes \bar{\mu}$ 
	has also a density which satisfies $f_{\bar{\xi}} \leq f_{\X}$.
\end{lemma}

\begin{proof}
	Boundedness is retained under
	the transformation to hyperspherical coordinates
	such that $f_{\xi^{\bm{U}}} \leq f_{\bm{U}(\X)}$.
	By integrating the angles $\tht$ out,
	this carries over to the marginal densities in the radius,
	$f_{R} \leq f_{R(\X)}$.
	Further, the distribution of $\X_{i}$ is invariant 
	with respect to $SO(d)$. 
	By the same arguments 
	as in Lemma~\ref{lemma:decomposition}, 
	the transformed vector $\bm{U}(\X_{i})$ has density
	$f_{\bm{U}(\X)}(\x) = f_{R(\X)}(r) f_{\bar{\mu}}(\tht)$,
	and the result follows. 
\end{proof}

For $d \geq 2$, 
let $\G \subset SO(d)$ be the finite group 
of rotations $\g$
which map the $d$-dimensional cross-polytope 
with vertices at the axes onto itself,
and let $\bar{\xi}_{\G} = \frac{1}{|\G|} \sum_{\g \in \G} \xi^{\g}$
be the symmetrization of $\xi$ with respect to $\G$.

\begin{lemma}
	\label{lemma:mean-of-rotated-designs}
	Let $\xi$ be invariant with respect to $\G$.
	Then
	\begin{equation}
		\label{eq:mean-of-rotated-designs}
		\M(\xi) 
		= \begin{pmatrix}
			\alpha & \bm{0}
			\\
			\bm{0} & \frac{1}{d} \int r^{2} \xi_{R}(\diff r) \Id_{d}
		\end{pmatrix} \, .
	\end{equation}
\end{lemma}

\begin{proof}
	When $\xi$ is invariant with respect to $\G$,
	then all components $x_{j}$ are invariant with respect to sign change,
	and any two components $x_{j}$ and $x_{j^{\prime}}$
	are exchangeable.
	Hence, the off-diagonal entries
	$\int x_{j} \xi(\diff \x)$ and $\int x_{j} x_{j^{\prime}} \xi(\diff \x)$, 
	$j \neq j^{\prime}$, are equal to zero,
	while all diagonal entries $\int x_{j}^{2} \xi(\diff \x)$
	 are equal to each other.
	 Further, 
	 $\sum_{j=1}^{d} \int x_{j}^{2} \xi(\diff \x)
	 	= \int R(\x)^{2} \xi(\diff \x)$,
	 	and the representation~\eqref{eq:mean-of-rotated-designs}
	 	follows
	 	\citep[cf.][Lemma 4.9.]{gaffke1996approximate}. 
\end{proof}

A design criterion $\Phi$ is invariant with respect to $SO(d)$
if $\Phi(\xi^{\g}) = \Phi(\xi)$ for any $\g \in SO(d)$ and any $\xi$.

\begin{theorem}
	\label{theorem:optimality-of-sym-design}
	For the multiple linear regression model with $d \geq 2$ covariates, 
	let $\Phi$ be a convex optimality criterion 
	that is invariant with respect to $SO(d)$. 
	Then for any design $\xi$ it holds that 
	\begin{equation*}
		\label{eq:optimality-of-sym-design}
		\Phi(\bar\xi) \leq \Phi(\xi) \, ,
	\end{equation*} 
	where $\bar{\xi} = \xi_{R} \otimes \bar{\mu}$
	is the symmetrization of $\xi$ with respect to $SO(d)$
	and $\xi_{R}$ is the marginal design of $\xi$ on the radius $r$.
\end{theorem}

\begin{proof}
	By the convexity of $\Phi$, we have
	\begin{equation}
		\label{eq:convexity}
		\Phi(\bar{\xi}_{\G}) \leq \frac{1}{|\G|} \sum_{\g \in \G} \Phi(\xi^{\g}) \, .
	\end{equation}
	Because $\G \subset SO(d)$,	
	$\Phi$ is invariant with respect to $\G$
	and, hence, $\Phi(\xi^{\g}) = \Phi(\xi)$ for all $\g \in \G$.
	As a consequence, the right hand side of
	the inequality~\eqref{eq:convexity} equals $\Phi(\xi)$.
	Further,
	notice that both $\bar{\xi}$ 
	and $\bar{\xi}_{\G}$ 
	have marginal $\xi_{R}$ on the radius
	and are invariant with respect to $\G$.
	Thus, the left hand side of
	the inequality~\eqref{eq:convexity} equals  
	$\Phi(\bar{\xi})$ by 
	Lemma \ref{lemma:mean-of-rotated-designs}. 
\end{proof}

\begin{proof}[Proof of Theorem \ref{theorem:dens-of-opt-design-spherical}]
	By Lemma~\ref{lemma:sym-design-bounded} 
	and Theorem~\ref{theorem:optimality-of-sym-design},
	we may restrict our search for a $D$-optimal subsampling design
	to the essentially complete class of invariant designs $\bar{\xi}$.
	By symmetry considerations,
	$\m(\bar{\xi}) = \mathbf{0}$
	and $\mathbf{S}(\bar{\xi})$ is a multiple of the identity matrix.
	Hence, the result follows from
	Theorem~\ref{theorem:opt-design-general-distribution}.
\end{proof}

The particular shape of the $D$-optimal subsampling design 
ensures that $\xi_{\alpha}^{*}$ is unique.	

\begin{proof}[Proof of equation~\eqref{eq:standard-normal-second-moment-optimal}]
	As in Lemma~\ref{lemma:mean-of-rotated-designs},
	we see that the information matrix of $\xi_{\alpha}^{*}$
	is of the form
	\[
		\M(\xi_{\alpha}^{*})
		= \begin{pmatrix}
				\alpha & \mathbf{0}
				\\
				\mathbf{0} 
					& m_{2}(\xi_{\alpha}^{*}) \Id_{d}
		\end{pmatrix} \, ,
	\]
	where $m_{2}(\xi_{\alpha}^{*}) 
		= \E\left[R^{2} \1_{\{R^{2} \geq \chi^{2}_{d, 1-\alpha}\}}\right] / d$
	by Theorem~\ref{theorem:dens-of-opt-design-spherical}.
	The squared radius $W = R^{2}$ has a $\chi^{2}$-distribution with $d$ degrees of freedom. 
	The truncated moment 
	$\E\left[W \1_{\{W \geq \chi^{2}_{d, 1-\alpha}\}}\right]
		= \int_{\chi^{2}_{d, 1-\alpha} }^{\infty} w f_{\chi^{2}_{d}}(w) \diff w$
	can be calculated by using the density
	$f_{\chi^{2}_{d}}(w) =  2^{-d/2} \Gamma(d/2)^{-1} w^{(d/2) - 1} \exp(-w/2)$
	of the $\chi^{2}$-distribution.
	Integration by parts yields
	\[
		m_{2}(\xi_{\alpha}^{*}) 
		= \frac{(\chi^{2}_{d, 1-\alpha})^{d/2} \exp(- \chi^{2}_{d, 1-\alpha} / 2)}%
					{d 2^{(d/2) - 1} \Gamma(d/2)} 
				+ \int_{\chi^{2}_{d, 1-\alpha}}^{\infty} \frac{w^{(d/2)-1} \exp(-w/2)}%
					{2^{d/2} \Gamma(d/2)} \diff w \, . 
	\]
	The first term on the right hand side can be written as
	$2 \chi^{2}_{d, 1-\alpha} f_{\chi^{2}_{d}}(\chi^{2}_{d, 1-\alpha}) / d$,
	while the second term simplifies to $\alpha$ 
	as the expression under the integral is the density 
	of the $\chi^{2}$-distribution. 
\end{proof}

\begin{lemma}
	\label{lemma:opt-transformed-design}
	Let the covariates $\X_{i}$ have density $f_{\X}$,
	let $\bm{A}$ be nonsingular,
	and let $\Z_{i} = \bm{A} \X_{i} + \mx$ be
	affine linearly transformed covariates.
	If $\xi_{\alpha}^{*}$ is a $D$-optimal subsampling design
	for the covariates $\X_{i}$, 
	then the transformed design 
	$\zeta_{\alpha}^{*} = (\xi_{\alpha}^{*})^{\g_{\bm{A}, \mx}}$ 
	is a $D$-optimal subsampling design for the covariates $Z_{i}$
	with density $f_{\Z}(z) = f_{\X}(\bm{A}^{-1}(\z - \mx)) / |\det(\bm{A})|$.
\end{lemma}

\begin{proof}
	First, note that 
	$\zeta_{\alpha} = \xi_{\alpha}^{\g_{\bm{A}, \mx}}$
	is a subsampling design for covariates $\Z_{i}$
	if and only if
	$\xi_{\alpha}$
	is a subsampling design for covariates $\X_{i}$.
	Further, by considerations of equivariance,
	$\M(\zeta_{\alpha}) 
		= \Q_{\bm{A}, \bm{\mu}} \M(\xi_{\alpha}) \Q_{\bm{A}, \bm{\mu}}^{\top}$
	and, hence,
	$\det(\M(\zeta_{\alpha})) 
		= \det(\Q_{\bm{A}, \bm{\mu}})^{2} \det(\M(\xi_{\alpha}))$.
	Thus, $\zeta_{\alpha}^{*}$ is $D$-optimal
	if and only if 
	$\xi_{\alpha}^{*}$ is $D$-optimal.
\end{proof}

\begin{proof}[Proof of Theorem \ref{theorem:dens-of-opt-design-elliptical}]
	Let $\bm{A}$ be a square root of $\bm{\Sigma}$,
	i.\,e.\ $\bm{A} \bm{A}^{\top} = \bm{\Sigma}$.
	If the distribution of $\X_{i}$ is elliptical with mean $\mx$ 
	and covariance matrix $\bm{\Sigma}$,
	then the distribution of $\bm{A}^{-1} (\X_{i} - \mx)$
	is spherical and centered.
	By Lemma~\ref{lemma:opt-transformed-design},
	the result follows.
\end{proof}

\begin{proof}[proof of Lemma~\ref{lemma:unbounded-distribution}]
	Let $f_{W}$ be the density of $W_{i} = R(\X_{i})^{2}$.
	Then 
	$d s^{2}(\xi_{\alpha}^{*}) 
		= \int_{q_{1 - \alpha}}^{\infty} w f_{W}(w) \diff w
		\geq \int_{q_{1 - \alpha}}^{\infty} q_{1 - \alpha} f_{W}(w) \diff w
		= \alpha q_{1 - \alpha}$.
	Hence,
	$s^{2}(\xi_{\alpha}^{*}) / \alpha \geq q_{1 - \alpha} / d$,
	and the right hand side tends to infinity 	for $\alpha \to 0$ 
	when the distribution of $\X_{i}$ is unbounded.
\end{proof}

\bibliographystyle{plainnat}  
\bibliography{ref_Linear}

\end{document}